\documentclass[11pt,a4paper]{article}
\usepackage[a4paper, top=2cm, bottom=2cm, left=1.8cm, right=1.8cm]{geometry}
\usepackage[toc,page]{appendix}  % dans le préambule

\usepackage{setspace}
\usepackage{fancyhdr}{\setlength{\headheight}{13.59999pt}} 
\usepackage{natbib}
\setcitestyle{authoryear,open={(},close={)}}
\usepackage{hyperref}
\usepackage{xcolor}
\usepackage[utf8]{inputenc}
\usepackage[T1]{fontenc}
\usepackage{lmodern}
\usepackage[french,english]{babel}
\usepackage{amsmath, amssymb, amsfonts, amsthm}
\usepackage{graphicx}
\usepackage{caption}
\usepackage{subcaption}
\usepackage{adjustbox}
\usepackage{verbatim}
\usepackage{booktabs}
\usepackage{diagbox}
\usepackage{makecell}
\usepackage{dsfont}
\usepackage{hyperref}

\newtheorem{lemma}{Lemma}

\newtheorem{assumptions}{Assumption}

\usepackage{bm} % pour \bm (alternative à \boldsymbol)
\newcommand{\btheta}{\bm{\theta}}

\newcommand{\ones}{\bm{1}}

\newcommand{\bF}{\bm{F}}
\newcommand{\bQ}{\bm{Q}}
\newcommand{\bI}{\bm{I}}

\newcommand{\bG}{\bm{G}}
\newcommand{\bH}{\bm{H}}
\newcommand{\bT}{\bm{T}}
\newcommand{\bR}{\bm{R}}
\newcommand{\ba}{\bm{a}}
\newcommand{\bU}{\bm{U}}
\newcommand{\bS}{\bm{S}}
\newcommand{\bs}{\bm{s}}
\newcommand{\bt}{\bm{t}}
\newcommand{\bu}{\bm{u}}

\newcommand{\bGamma}{\bm{\Gamma}}
\newcommand{\bmu}{\bm{\mu}}
\newcommand{\bZ}{\bm{Z}}

\newcommand{\bX}{\bm{X}}

\newcommand{\bV}{\bm{V}}
\newcommand{\bx}{\bm{x}}

\newcommand{\bpi}{\bm{\pi}}

\newcommand{\brho}{\bm{\rho}}
\newcommand{\bbeta}{\bm{\eta}}

\newcommand{\keywords}[1]{%
  \vspace{1em}
  \noindent\textbf{Keywords:} #1
}
\newenvironment{funding}
  {\section*{Funding}}
  {}
\newenvironment{Conflicts of interest}
  {\section*{Conflicts of interest}}
  {}

\title{Model-Based Clustering of Football Event Sequences:\\ A Marked Spatio-Temporal Point Process Mixture Approach}

\author{
  Koffi Amezouwui\textsuperscript{1}\thanks{Corresponding author: \href{mailto:koffi.amezouwui@ensai.fr}{koffi.amezouwui@ensai.fr}},
  Brigitte Gelein\textsuperscript{2},
  Matthieu Marbac\textsuperscript{3},
  Anthony Sorel\textsuperscript{4}
}

\date{
  \textsuperscript{1}Univ. Rennes, Ensai, CNRS, CREST-UMR 9194, 35000 Rennes, France \\
  \textsuperscript{2}Ensai, IRMAR-UMR, CNRS 6625, France \\
  \textsuperscript{3}Université Bretagne Sud, UMR CNRS 6205, LMBA, F-56000 Vannes, France \\
  \textsuperscript{4}Univ. Rennes, Inria, M2S, 35000 Rennes, France \\
  \vspace{0.5em}
}

\begin{document}

\maketitle

\abstract{We propose a novel mixture model for football event data that clusters entire possessions to reveal their temporal, sequential, and spatial structure. Each mixture component models possessions as marked spatio-temporal point processes: event types follow a finite Markov chain with an absorbing state for ball loss, event times follow a conditional Gamma process to account for dispersion, and spatial locations evolve via truncated Brownian motion. To aid interpretation, we derive summary indicators from model parameters capturing possession speed, number of events, and spatial dynamics. Parameters are estimated through maximum likelihood via Generalized Expectation-Maximization algorithm. Applied to StatsBomb data from 38 Ligue 1 matches (2020/2021), our approach uncovers distinct defensive possession patterns faced by Stade Rennais. Unlike previous approaches focusing on individual events, our mixture structure enables principled clustering of full possessions, supporting tactical analysis and the future development of realistic virtual training environments.
}

\keywords{clustering, event data, marked point processes, Markov chains, mixture models, soccer}

\section{Introduction}

Recent advances in training technologies have led professional soccer clubs to adopt virtual environments \citep{richlan2023virtual,witte2025sports} that enable targeted practice with reduced physical strain, helping to prevent injuries and fatigue  \citep{jia2024application,demeco2024role,cariati2025virtual}.  The effective use of these technologies requires tools capable of populating the environments with realistic and contextually relevant scenarios, usually derived from data collected during actual matches. To inform the design of training scenarios that focus on specific situations, it is therefore essential to analyze play at the level of team possessions. In this context, it becomes necessary to identify subsets of possessions that are relevant to a specific training objective, which motivates the use of clustering. By grouping possessions into clusters representing different tactical or situational contexts, a coach can ensure that a virtual training environment is populated only with possessions that align with the goal of a given session. For instance, when a particular possession in a match results in a failure or undesired outcome, the cluster corresponding to that possession can be identified, and other possessions from the same cluster can be used to design realistic exercises that allow players to rehearse and improve their decision-making in comparable scenarios.

A natural statistical approach for such clustering is provided by finite mixture models, which represent the distribution of observed data as a combination of component distributions, with clusters defined by shared component membership \citep{McL00,Fruhwirth2019handbook}. In this framework, a standard definition of a cluster corresponds to the subset of individuals generated by the same mixture component (see \citealp{hennig2010methods} and \citealp{baudry2010combining} for several extensions, and \citealp{hennig2015true} for a discussion on cluster definitions). A finite mixture model is characterized by three main elements: the number of mixture components, the mixing proportions, and the component-specific distributions. Using such models to analyze sports data generally requires methodological developments to account for the specificity of the data as well as the clustering objective (see \citealp{yin2023analysis} and \citealp{scrucca2025model} for applications in basketball, \citealp{leroy2018functional} and \citealp{bouvet2024investigating} for applications in swimming, or \citealp{sawczuk2024bayesian} for applications in rugby). Implementing these models requires careful characterization and understanding of the collected variables to select appropriate component-specific distributions. In sports, data are primarily of two types: movement (tracking) data and event data. Movement data consist of time-stamped locations tracking all players and the ball, typically captured via optical systems (see \citealp{santos2022role}, \citealp{Sandholtz}, and \citealp{kovalchik2023player} for statistical developments for movement data). Event data, by contrast, record sequences of game events and are collected manually through video annotation (see \citealp{Narayanan}, \citealp{qi2024made}, and \citealp{jensen2009bayesball} for statistical developments for event data). Although event data are less dense than movement data, they contain sufficient information to characterize possessions, notably through key events, their locations on the field, and the speed at which they occur. Moreover, while modeling these data for clustering is complex, it is still simpler than for movement data and requires less storage space. Accordingly, we focus on event data and extend finite mixture models to handle their spatio-temporal structure, treating possessions as realizations of point processes rather than conventional multivariate observations.

In this paper, we aim to provide a clustering method to populate a virtual training environment for Stade Rennais goalkeepers with different groups of opponent possessions. For this purpose, we use event data provided by StatsBomb from 38 matches of the 2020/2021 French Ligue 1 season, focusing exclusively on possessions where Rennes is in defense. These event data can be seen as marked spatio-temporal point processes, where each on-ball action by Rennes’ opponents is a point in time and space, and the mark corresponds to the categorical event type, including absorbing states that determine possession loss (e.g., shot, foul). Temporal variables capture the exact timestamp of each action to reflect possession dynamics, while spatial coordinates, constrained to the convex pitch area, represent the location of the action. Each possession ends when an event corresponding to an absorbing state occurs, explicitly indicating the loss of the ball, which allows us to segment the match into meaningful possession for analysis and clustering.

Clustering methods applied to football possessions have primarily focused on segmenting and analyzing phases of play using event data that detail every on-field action (passes, shots, ball losses, etc.). Classical approaches often rely on non-model-based  unsupervised techniques   representing each possession with aggregated feature vectors, including the number of actions, possession duration, or initial and final player positions \citep{narizuka2019clustering,fernandez2024reporting}. Other studies leverage finer temporal and sequential structures by modeling possessions as marked event sequences and applying similarity measures on action trajectories or passing networks %. %Some methods combine clustering with probabilistic models, such as mixture models or Hidden Markov Models, to capture recurring play patterns while accounting for temporal and spatial variations within possessions 
\citep[Chapter 6]{el2021game}. However, these approaches remain limited: they rarely capture the temporal, sequential, and spatial dimensions of possessions simultaneously and often rely on aggregated representations or ad hoc similarity metrics, restricting tactical interpretability and generalization to full-game simulation. These limitations motivate the development of a marked spatio-temporal point process mixture model, capable of coherently capturing the complete structure of football possessions.

To develop the mixture model needed for the analysis of football event data, we build on existing statistical models designed for this type of data. Early statistical works in football analytics, such as \citet{dixon1997modelling} and more recently \citet{michels2025extending}, focused on modeling match scores and outcomes using probabilistic frameworks, highlighting the benefits of hierarchical and time-dependent structures for predictive accuracy. Building specifically on event-level data, \citet{baio2010bayesian} proposed a Bayesian hierarchical model to predict goals while accounting for team interactions and match context. More recently, \citet{Narayanan} developed a flexible framework treating on-ball actions as marked spatio-temporal point processes, where the type of event constitutes the mark.  By decoupling the modeling of event times from event types, the approach allows flexible specification for the marks while avoiding clustering of event times. Using a comprehensive Bayesian framework, it captures team-specific abilities, event interactions across locations, and game dynamics, and supports prediction and simulation of event sequences, even with limited data. The model also provides interpretable parameters, such as team ability rankings by event type and conversion rates, which offer insights into the underlying structure of the game. However, unlike in our context, their methodology does not allow for clustering, since it is not a mixture model, and it treats individual events as the statistical unit. In contrast, we focus on entire possessions sequences of events ending with an absorbing state corresponding to ball loss allowing us to cluster possessions into distinct defensive scenarios. This clustering is crucial for populating a virtual training environment for Rennes goalkeepers with realistic possession sequences.

We propose a mixture model tailored to the structure of football possessions. Each mixture component models the vector of event types within a possession using a finite Markov chain with an absorbing state, which represents ball loss and therefore the end of the possession. Conditional on the sequence of event types, the occurrence times of events are modeled using a conditional Gamma process, capturing both under- and over-dispersion while avoiding auto-excitation, in line with the observations of \citet{Narayanan}. Finally, for each component, the spatial locations of events are modeled using a Brownian motion truncated to the convex space of the pitch, conditional on their type and occurrence time. This modeling framework captures the main characteristics of the data while allowing an intuitive analysis of each cluster in terms of event sequencing speed within a possession, the number of events per possession, and the spatial dynamics. To further facilitate interpretation, we introduce three summary indicators derived from the model parameters, which encapsulate these key features for each cluster. Applied to StatsBomb event data, our mixture model identifies representative possession types that capture realistic game dynamics in terms of tempo, event sequences, and spatial patterns. This modeling framework provides the essential clustering of possessions, which can later be used to populate a virtual training environment for Rennes goalkeepers with varied and plausible scenarios.

This paper is organized as follows. Section~\ref{data::description} introduces the context of the study and the StatsBomb dataset. Section~\ref{sec::model} presents a mixture model of marked point processes with an absorbing mark. Section~\ref{sec:MLE} describes the GEM algorithm used for maximum likelihood estimation. A numerical study of this approach on simulated data is provided in Section~\ref{sec:simu}. Section~\ref{sec:analysis} is devoted to the analysis of the dataset and the interpretation of the clusters obtained with the proposed method. Finally, Section~\ref{sec:concl} concludes with key insights and directions for future research.

\section{StatsBomb data description}\label{data::description}
\subsection{Dataset overview}
The data used in this study are provided by StatsBomb and come from 38 professional football matches played during the 2020/2021 Ligue 1 season in France. The Ligue 1 championship consists of 20 teams playing in a double round-robin format, whereby each team faces all others twice, once at home and once away, resulting in 38 matches per team throughout the season. Each match consists of two 45-minute halves, separated by a 15-minute half-time interval. We focus exclusively on the matches played by Stade Rennais Football Club (Rennes) and have access to event data, which record discrete on-ball actions occurring during the match, such as passes, receptions, and shots, each annotated with a timestamp and spatial coordinates on the pitch. We aim to model the offensive phases of Rennes’ opponents in order to train the Rennes defensive players. This means analyzing only those possessions in which Rennes opponents have the ball. A possession is a phase of uninterrupted succession of events that ends with the loss of the ball due to a specific event (e.g., foul, interception, or the ball going out of play). Moreover, events characterize what happens in a possession. It is therefore natural to be interested in this type of data to cluster the different possessions. The event data capture player actions, enabling analysts, coaches, and researchers to perform an in-depth analysis of match dynamics and team behaviors. To ensure a reliable analysis, the raw data were preprocessed to remove anomalies, as detailed in Subsection \ref{data::prep}. This process corrected various anomalies, such as timestamp inconsistencies and overlapping events. For each match, the collected data includes several variables, such as the match period, event type, possession team name, event timestamp with millisecond precision, and the spatial coordinates of the event expressed in meters relative to the pitch. This dataset thus provides sufficient granularity to enable detailed analysis of possession sequences, both individual and collective actions, and the spatio-temporal behavior of teams throughout the matches.

\subsection{Data preprocessing}\label{data::prep}
The data preparation process involved several steps to ensure the raw event data was ready for analysis. The event coordinates, recorded on a pitch of size $[0,120]\times[0,80]$, are initially expressed in the reference frame of the team performing the action of event type (always oriented left to right).

Prior to analysis, several inconsistencies were identified in the original dataset, the most critical involving duplicate events with identical timestamps and non-uniform timestamp formats. These issues were corrected before proceeding. To ensure spatial coherence, all event coordinates were projected onto a common reference frame. Missing values were imputed by carrying forward the coordinates from the most recent valid event. Approximately $4\%$ of the possessions were discarded due to manifest labeling errors, and additional possessions were excluded because of unrealistic player movement, defined as a velocity between consecutive events exceeding the $98^{th}$ percentile. After these preprocessing steps, the final dataset used for analysis comprised $n = 2623$ possessions. The number of possessions per matchday ranges from 48 to 112, with an average of 69.03.

\subsection{Structure and spatio-temporal dynamics of possessions}
 %To track temporal dynamics within possessions, a new variable was constructed measuring the relative timing of each event with respect to the begining of the possession.  The number of possessions per matchday ranges from 48 to 112 with an average of 69.03.  

For analytical purposes, a possession systematically ends with a ball loss, recorded as a unique event labeled \emph{End Possession}. Each new possession begins immediately with the recovery of the ball, labeled \emph{Start Possession}. Table~\ref{tab:event_occurrences} summarizes the number of occurrences for each event type across the cleaned dataset. Passes and ball carries account for the majority of actions, followed by pressures and possession boundaries. Less frequent but still relevant events include duels, blocks, shots, and various defensive interventions such as interceptions and ball recoveries.

\begin{table}[h]
\centering
\begin{tabular}{l r @{\hskip 1.5cm} l r}
\toprule
\textbf{Event Type} & \textbf{Count} & \textbf{Event Type} & \textbf{Count} \\
\midrule
Pass               & 10780 & Miscontrol       & 184   \\
Carries            & 7753  & Block            & 180   \\
Pressure           & 2986  & Shot             & 166   \\
Start Possession   & 2623  & Ball Recovery    & 120   \\
End Possession     & 2623  & Interception     & 89    \\
Ball Receipt      & 541     & Goal Keeper      & 51    \\
Clearance          & 360     & Foul Won         & 23    \\
Duel               & 185     &                  &       \\
\bottomrule
\end{tabular}
\caption{
Number of occurrences for each event type across the 2623 possessions kept for the analysis, presented in decreasing order.}
\label{tab:event_occurrences}
\end{table}

\begin{comment}
\begin{table}[h]
\centering
\begin{tabular}{l r @{\hskip 1.5cm} l r @{\hskip 1.5cm} r}
\toprule
\textbf{Event Type} & \textbf{Count} & \textbf{Event Type} & \textbf{Count} & \textbf{Total} \\
\midrule
Pass              & 10\,456 & def\_Pass              & 582  & 11\,038 \\
Carries           & 7\,644  & def\_Carries           & 323  & 7\,967  \\
Start Possession  & 2\,596  & def\_Start Possession  & 89   & 2\,685  \\
End Possession    & 1\,849  & def\_End Possession    & 836  & 2\,685  \\
Pressure          & 374     & def\_Pressure          & 2\,730& 3\,104  \\
Ball Receipt     & 538     & def\_Ball Receipt     & 20   & 558    \\
Shot              & 169     & def\_Shot              & 2    & 171    \\
Miscontrol        & 149     & def\_Miscontrol        & 38   & 187    \\
Ball Recovery     & 77      & def\_Ball Recovery     & 49   & 126    \\
Block             & 32      & def\_Block             & 155  & 187    \\
Duel              & 31      & def\_Duel              & 164  & 195    \\
Clearance         & 25      & def\_Clearance         & 347  & 372    \\
Interception      & 20      & def\_Interception      & 72   & 92     \\
Foul Won          & 22      & def\_Foul Won          & 3    & 25     \\
Goal Keeper       & 9       & def\_Goal Keeper       & 48   & 57     \\
\bottomrule
\end{tabular}
\caption{Number of occurrences for each event type and their defensive counterparts across the 2\,685 possessions kept for analysis.}
\label{tab:event_occurences}
\end{table}
\end{comment}

Our dataset also includes variables capturing the time differences between consecutive events, as well as the coordinates of the events, enabling a detailed analysis of game dynamics. On average, each possession comprises 10.93 events, with a mean inter-event time of 2.49 seconds and a variance of 41.06. The average elapsed time until possession loss or termination (\emph{i.e.}, the time to reach the \emph{End Possession} state) is approximately 27.25 seconds. These metrics are essential for characterizing the nature of each possession, as they reflect both the temporal rhythm and the spatial progression of play.

We analyzed the transition dynamics between event types by examining how frequently one type of event follows another within the flow of possession. Each transition corresponds to a change from a preceding event to a current event, allowing us to assess structured patterns in the sequencing of actions. Pearson’s $\chi^2$ test of independence revealed a highly significant dependence between current and preceding event types $(p < 2.2 \times 10^{-16})$. This result confirms that event transitions follow certain patterns. Leveraging the joint availability of timestamps and event types, we examined whether inter-event times depend on sequential structure. A two-way analysis of variance (ANOVA) was performed to assess the effects of both preceding and current event types, with results indicating that both factors exert a highly significant influence ($p < 2 \times 10^{-16}$). However, the preceding event accounts for a substantially larger share of the variance, indicating that actions occurring immediately beforehand exert a strong influence on the pace of subsequent play. To further investigate spatial dynamics, we analyzed horizontal and vertical displacements between consecutive events. Descriptive statistics of these displacements were computed in the same manner as for inter-event times, enabling a detailed examination of movement patterns across different sequences of play. This analysis reveals directional patterns associated with different sequences of play, while ensuring that all movements remain confined within the boundaries of the pitch.

\section{Model-based clustering of marked point processes with absorbing mark}\label{sec::model}
	
\subsection{Random variables involved by soccer event data}
	We consider a random sample $\mathbf{X}=(\bX_1^\top,\ldots,\bX_{n}^\top)^\top$ composed of $n$ independent and identically distributed sequences of events, where $\bX_i$ describes the events of sequence $i$. A sequence ends with the loss of ball possession by a team, and thus terminates with a specific event that causes this loss (\emph{e.g.,} a foul, interception by a player of the opposing team, the ball being sent out of the field, etc.). As soon as a sequence ends, the next sequence begins. Sequence $i$ is described by $M_i$ events leading that $\bX_i=(\bX_{i,0}^\top, \ldots, \bX_{i, M_i}^\top)^\top$.
Event $j$ of possession $i$ is described by a triplet of random variables, such that $\bX_{i,j}=(S_{i,j},T_{i,j}, \bU_{i}^\top)^\top$, where $S_{i,j} \in \mathcal{E}$ indicates the nature of the event, $T_{i,j} \in \mathbb{R}^{+}$ denotes the time of occurrence of event $j$ since the beginning of possession $i$  and $\bU_{i}=(U_{i,j,1},U_{i,j,2})^\top\in[b_{1,1},b_{2,1}]\times[b_{1,2},b_{2,2}]$ indicates the coordinates of the events in the field. The set $\mathcal{E} = {0, \ldots, E+1}$ represents the different event types. Event 0 denotes the acquisition of ball possession by the team, whereas Event $E+1$ corresponds to the loss of possession. The intermediate events $1$ to $E$ do not entail any change in possession status. Consequently, each possession sequence necessarily begins with Event 0, which does not reoccur within the same sequence, and terminates upon the occurrence of Event $E+1$. Therefore, for any possession $i=1,\ldots,n$, we have $S_{i,0}=0$, $T_{i,0}=0$, $S_{i,M_i}=E+1$ and $S_{i,j}\notin\{0,E+1\}$ for $j=1,\ldots,M_i-1$. In addition, the possession $i$ starts where possession $i-1$ ends leading that $\bU_{i,0}=\bU_{i-1,M_{i-1}}$. The duration of a sequence $T_{i,M_i}$ as well as the number of events $M_i$ are random and so may vary between the $n$ sequences. Hence, $\bX_i$ is a marked point process where the marks are given by $\bS_i=(S_{i,1},\ldots,S_{i,M_i})^\top$ and $\bU_i=(\bU_{i,1}^\top,\ldots,\bU_{i,M_i}^\top)^\top$. This process ends when the mark providing the nature of the event reaches an absorbing state that indicates the end of the possession.

Table \ref{tab:data_clean_sequence} presents a snapshot of the cleaned event data for the 38 matches (specifically from Matchday 1). Each row in the table represents a transition between two consecutive events within a single possession. It includes variables such as the previous and current event types ($S_{i,j-1}$ and $S_{i,j}$), the elapsed time since the beginning of the possession ($T_{i,j}$), the spatial coordinates of the events in the field ($U_{i,j,1}$, $U_{i,j,2}$). The labelled events, i.e., the Start Possession and the End Possession, enable each phase of possession to be clearly segmented, making it easier to model the dynamics of sequential events. To illustrate the structure and flow of a single possession more clearly, Figure~\ref{fig:poss5_matchday1} presents an example sequence taken from the possession shown in Table~\ref{tab:data_clean_sequence}, based on cleaned and labeled data from Matchday 1. This visualization highlights the temporal and spatial chaining of events, offering insights into how play unfolds within a single possession.

\begin{comment}
\begin{table}[h]
\centering
\begin{tabular}{rllrrr}
\toprule
  & $S_{i,j-1}$ & $S_{i,j}$ & \textbf{$T_{i,j}$} & \textbf{$U_{i,j,1}$} & \textbf{$U_{i,j,2}$} \\ 
  \hline
  &Start Possession & Carries & 0.45 & 13.10 & 50.50 \\ 
&Carries & Pass & 1.13 & 15.50 & 56.30 \\ 
&Pass & Carries & 1.77 & 21.30 & 58.40 \\ 
& Carries & Pressure & 2.14 & 89.10 & 21.10 \\ 
& Pressure & Duel & 2.78 & 95.10 & 20.30 \\ 
& Duel & Pressure & 3.52 & 30.10 & 69.80 \\ 
& Pressure & End Possession & 5.33 & 82.60 & 5.00 \\ 
\bottomrule
\end{tabular} 
\caption{Events sequence for possession 95 on Matchday 1, showing transitions alongside temporal ($T_{i,j}$) and spatial ($U_{i,j,1}$, $U_{i,j,2}$) coordinates on the pitch.}
\label{tab:data_clean_sequence}
\end{table}
\end{comment}

%\begin{figure}[h]
%    \centering
%\includegraphics[width=\textwidth, keepaspectratio]{Poss_95_M1.PNG}
%\caption{Spatial and temporal representation of events during possession 95 on Matchday 1, showing the sequence of actions along with their field locations and time progression.}
%\label{fig:possession95} 
%\end{figure} 

\begin{table}[h]
\centering
\begin{tabular}{rllrrr}
\toprule
  & $S_{i,j-1}$ & $S_{i,j}$ & \textbf{$T_{i,j}$} & \textbf{$U_{i,j,1}$} & \textbf{$U_{i,j,2}$} \\ 
  \hline
  & Start Possession & Pass & 11.9 & 77.9 & 5.1 \\ 
  & Pass & Carries & 14.0 & 53.9 & 18.2 \\ 
  & Carries & Pass & 15.5 & 53.7 & 18.8 \\ 
  & Pass & Carries & 16.5 & 63.9 & 21.5 \\ 
  & Carries & Pass & 16.5 & 64.2 & 21.5 \\ 
  & Pass & Carries & 18.0 & 47.0 & 21.6 \\ 
  & Carries & Pass & 19.0 & 46.9 & 21.6 \\ 
  & Pass & Carries & 20.6 & 39.7 & 58.1 \\ 
  & Carries & Pressure & 22.1 & 44.1 & 61.9 \\ 
  & Pressure & Pass & 22.3 & 45.6 & 62.9 \\ 
  & Pass & Carries & 23.7 & 50.0 & 77.0 \\ 
  & Carries & Pass & 24.3 & 50.0 & 77.0 \\ 
  & Pass & Carries & 25.1 & 62.6 & 78.2 \\ 
  & Carries & Pass & 25.2 & 62.6 & 78.2 \\ 
  & Pass & Carries & 26.4 & 66.9 & 70.7 \\ 
  & Carries & Foul Won & 27.2 & 70.3 & 72.3 \\ 
  & Foul Won & Pass & 29.8 & 85.4 & 74.0 \\ 
  & Pass & End Possession & 30.1 & 85.8 & 66.4 \\ 
\bottomrule
\end{tabular} 
\caption{Events sequence for possession 5 on Matchday 1, showing transitions alongside temporal ($T_{i,j}$) and spatial coordinates ($U_{i,j,1}$, $U_{i,j,2}$) on the pitch.}
\label{tab:data_clean_sequence}
\end{table}

\begin{comment}
\begin{table}[ht]
\centering
\begin{tabular}{lccc}
  \hline
Event type & Time (sec) & Location\_x (m) & Location\_y (m) \\ 
  \hline
Start Possession & 0.00 & 77.90 & 5.10 \\ 
  Pass & 11.88 & 77.90 & 5.10 \\ 
  Carries & 14.02 & 53.90 & 18.20 \\ 
  Pass & 15.48 & 53.70 & 18.80 \\ 
  Carries & 16.50 & 63.90 & 21.50 \\ 
  Pass & 16.51 & 64.20 & 21.50 \\ 
  Carries & 18.05 & 47.00 & 21.60 \\ 
  Pass & 19.03 & 46.90 & 21.60 \\ 
  Carries & 20.63 & 39.70 & 58.10 \\ 
  Pressure & 22.05 & 44.10 & 61.90 \\ 
  Pass & 22.28 & 45.60 & 62.90 \\ 
  Carries & 23.70 & 50.00 & 77.00 \\ 
  Pass & 24.31 & 50.00 & 77.00 \\ 
  Carries & 25.12 & 62.60 & 78.20 \\ 
  Pass & 25.24 & 62.60 & 78.20 \\ 
  Carries & 26.37 & 66.90 & 70.70 \\ 
  Foul Won & 27.23 & 70.30 & 72.30 \\ 
  Pass & 29.83 & 85.40 & 74.00 \\ 
  End Possession & 30.11 & 85.80 & 66.40 \\ 
   \hline
\end{tabular}
\end{table}
\end{comment}

\begin{figure}[h]
    \centering
\includegraphics[width=1.0\linewidth]{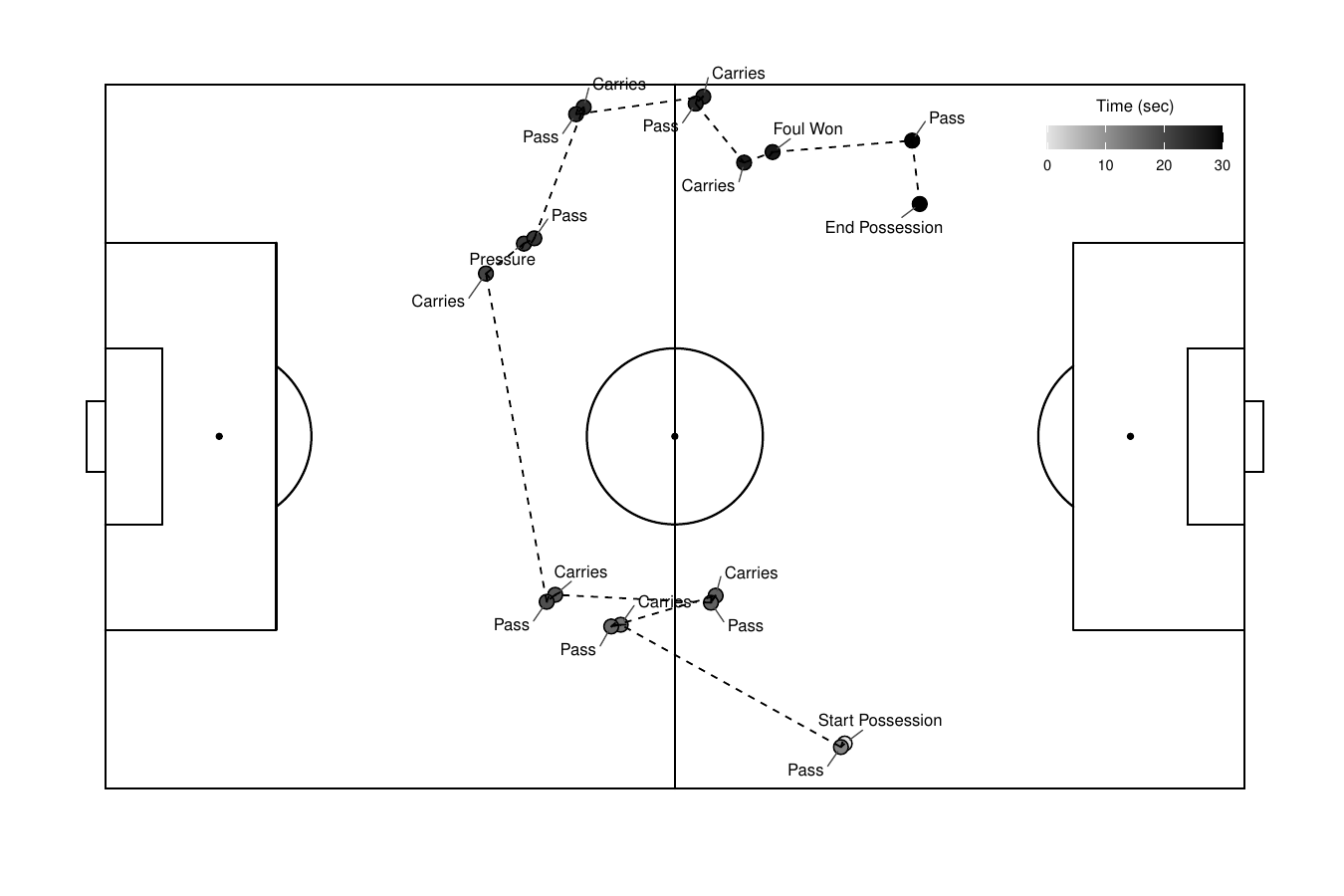}
    \caption{Spatial and temporal representation of events during possession 5 on Matchday 1, showing the sequence of actions along with their field locations and time progression. For visual clarity, Pass events were slightly shifted down-left and Carries events were shifted up-right by 0.4 units on both axes. This offset is purely for visualization and does not affect the original event coordinates.}
    \label{fig:poss5_matchday1}
\end{figure}

\subsection{Mixture model of point processes with Markovian properties}
We aim to cluster soccer sequences based on information describing their events (nature of the events, time of arrivals of the events, location of the events) in order to identify groups composed of similar chains of events. More specifically, a group consists of a subset of possessions whose sequences of events follow the same distribution. To this end, we consider a latent random variable $\bZ_i=\left(Z_{i, 1}, \ldots, Z_{i, K}\right)^\top$ that represents the cluster assignment for possession $i$, where $Z_{i, k}=1$ if possession $i$ belongs to cluster $k$, and $Z_{i, k}=0$ otherwise, with $K$ being the number of clusters.
The latent variables $\bZ_i$ are independent and follow a multinomial distribution parameterized by $\boldsymbol{\pi}=\left(\pi_1,\ldots,\pi_K\right)^\top$, where $\pi_k=\mathbb{P}(Z_{i, k}=1)$ denotes the probability that a sequence belongs to cluster $k$, and satisfies $\sum_{k=1}^K \pi_k=1$. 
Therefore, the marginal density of $\bX_i$ is defined by
\begin{equation}\label{eq:model1} 
f(\bx_i ; \btheta) = \sum_{k=1}^{K} \pi_k g(\bs_i;\bGamma_k)h(\bt_i\mid \bs_i; \boldsymbol{\rho}_k)\ell(\bu_i\mid\bs_i,\bt_i; \bbeta_k), 
\end{equation}
with $\btheta = \{\bpi, \bGamma_1,\ldots,\bGamma_K,\brho_1, \ldots, \brho_K,\bbeta_1,\ldots,\bbeta_K \}$ groups all the model parameters, $g(\cdot; \bGamma_k)$, $h(\cdot\mid \bs_i; \boldsymbol{\rho}_k)$ and $\ell(\cdot\mid\bs_i,\bt_i;\bbeta_k)$ denotes the conditional probability density function of $\bS_i\mid Z_{i,k}=1$,  $\bT_i\mid \bS_i=\bs_i,Z_{i,k}=1$ and  $\bU_i\mid\bS_i=\bs_i,\bT_i=\bt_i,Z_{i,k}=1$ respectively.

The model assumes that the vector of marks indicating the type of events, given that possession $i$ belongs to cluster $k$, follows a finite-state Markov chain defined by
\begin{equation}\label{eq:model2}  g(\bs_i;\bGamma_k)  = \prod_{j=1}^{M_i} \bGamma_k[s_{i,j-1},s_{i,j}],\end{equation}
where $\bGamma_k$ is an $(E+1) \times (E+1)$ stochastic matrix (\emph{i.e.,} each element is strictly positive and each row sums to one) such that $\bGamma_k[e,e']$ denotes the probability of observing event $e' \in \{1,\ldots,E+1\}$ given that the previous event was $e$. Although the cardinality of $\mathcal{E}$ is $E+2$, event $0$ never reoccurs within a possession and event $E+1$ corresponds to an absorbing state that terminates the process. Consequently, the finite Markov chain can be fully characterized by $\bGamma_k$. By completing this matrix with a column of zeros and a row of zeros except for the last element, which is equal to one, the transition matrix over all $E+2$ events is obtained.

The conditional distribution of the vector of arrival times, given the event types and the possession cluster, is assumed to follow a Gamma process, which implies that the inter-event times are conditionally independent given the event types and cluster memberships. Denoting $\Delta t_{i,j} = t_{i,j} - t_{i,j-1}$ as the time elapsed between event $j-1$ and event $j$ for possession $i$, the model further assumes that $\Delta t_{i,j}$ is conditionally independent of any other variable given the component membership and the type of the $j$-th event of the possession (\emph{i.e.}, given $(S_{i,j},\bZ_i^\top)^\top$). Hence, we have
\begin{equation}\label{eq:model3} h(\bt_i\mid \bs_i; \boldsymbol{\rho}_k ) = \prod_{j=1}^{M_i} \prod_{e=1}^{E+1} \left[\gamma(\Delta t_{i,j}; \rho_{k,e,1} , \rho_{k,e,2})\right]^{\mathds{1}_{\{S_{i,j}=e\}}},\end{equation}
where $\gamma(\cdot;a,b)$ denotes the density of a gamma distribution with shape parameter $a$ and scale parameter $b$, $\brho_k=(\brho_{k,1},\ldots,\brho_{k,E+1})^\top$ collects the parameters of all gamma distributions, and $\brho_{k,e}=(\rho_{k,e,1},\rho_{k,e,2})^\top$ are the parameters of the gamma distribution associated with event $e$ for component $k$.
    
The conditional distribution of the difference in spatial coordinates, given $\bZ_i$, $\bS_i$, and $\bT_i$, is modeled by a truncated Gaussian distribution centered at the previous location, with a diagonal variance matrix that depends on the component membership and the previous state. Specifically, the conditional distribution of $U_{i,j,h}$, given the history of the possession up to event $j$ and the cluster memberships, is a truncated Gaussian distribution that depends on $(\bZ_i, U_{i,j-1,h}, S_{i,j}, \Delta t_{i,j})$, with a support ensuring that the events occur within the field of dimensions $[b_{1,1}, b_{2,1}] \times [b_{1,2}, b_{2,2}]$, corresponding to a length of [0,120] and a width of [0,80]. In particular, the conditional distribution of $U_{i,j,h}$ is a truncated Gaussian on $[b_{1,h}, b_{2,h}]$, with mean $U_{i,j-1,h}$ and variance $\eta^2_{k,h,e}\Delta t_{i,j}$, where $e$ is such that $S_{i,j} = e$. Therefore, we have 
\begin{equation}\label{eq:model4} 
\ell(\bu_i\mid\bs_i,\bt_i; \bbeta_k) = \prod_{j=1}^{M_i}\prod_{h=1}^2\prod_{e=1}^{E+1} \left[\frac{ \phi(\Delta u_{i,j,h}/\eta_{k,h,e}) }{\eta_{k,h,e} \left[ \Phi\left( \frac{\Delta b_{2,i,j,h}}{\eta_{k,h,e}}\right) - \Phi\left(\frac{\Delta b_{1,i,j,h}}{\eta_{k,h,e}}\right) \right] } \mathds{1}_{\{b_{1,h}\leq u_{i,j,h}\leq b_{2,h}\}}\right]^{\mathds{1}_{\{S_{i,j}=e\}}},
\end{equation}
where $\bbeta_k=(\eta_{k,1,1},\eta_{k,2,1},\eta_{k,1,2},\ldots,\eta_{k,2,E})$, $\Delta u_{i,j,h}=(u_{i,j,h}-u_{i,j-1,h})/\sqrt{\Delta t_{i,j}}$, $\Delta b_{1,i,j,h}=(b_{1,h}-u_{i,j-1,h})/\sqrt{\Delta t_{i,j}}$,  $\Delta b_{2,i,j,h}=(b_{2,h}-u_{i,j-1,h})/\sqrt{\Delta t_{i,j}}$, $\phi$ is the density of a standard Gaussian distribution and $\Phi$ is the cumulative distribution function of a standard Gaussian distribution.

\subsection{Model properties}\label{subsec:prop}
In this section, we state some properties of the model defined by \eqref{eq:model1}-\eqref{eq:model4}. Establishing the identifiability of the model parameters is a fundamental prerequisite for ensuring the validity of parameter estimation procedure and, consequently, the reliability of the estimated partition. The following proposition gives sufficient conditions that ensure that this property holds.
\begin{assumptions}\label{ass:ID}
\begin{enumerate}
\item Each proportion is strictly positive: $\pi_k> 0$.\label{ass:prop} 
\item For any component $k$, the transition matrix of the finite Markov chain implies that $\{1,\ldots,E\}$ belong to the same communication class and that this class is aperiodic and transient.\label{ass:Gamma} 
\item The parameters stating the conditional gamma process are different: $\brho_{k,e}\neq\brho_{\ell,e}$ for $k\neq\ell$ and any $e\in\{1,\ldots,E+1\}$.\label{ass:rho}
\end{enumerate}
\end{assumptions}

\begin{lemma}\label{lem:idengrl}
If Assumption~\ref{ass:ID} holds true, then the parameters of the model defined by \eqref{eq:model1}-\eqref{eq:model4} are identifiable up to label swapping.
\end{lemma}

From the parameters of the model defined by \eqref{eq:model1}–\eqref{eq:model4}, we can compute three statistics that facilitate the interpretation of the clusters.
The first is the expected number of events in a possession, excluding the event corresponding to the gain of possession, given that the possession originates from component~$k$. It is defined by
$$
\lambda_k=\mathbb{E}[M_i\mid Z_{i,k}=1].
$$
The second is the conditional expectation of $\bV_i=(V_{i,1},\ldots,V_{i,E})$ where $V_{i,e}=\sum_{j=1}^{M_i} \mathds{1}_{S_{i,j}=e}$ corresponds to the number of occurrences of event~$e$, given that the possession originates from component~$k$. We denote this by $\boldsymbol{\kappa}_k = (\kappa_{k,1}, \ldots, \kappa_{k,E})^\top$, where
$$
\boldsymbol{\kappa}_{k} = \mathbb{E}[\bV_{i}\mid Z_{i,k}=1].
$$
The third is the conditional expectation of the duration of a possession, given that the possession originates from component~$k$, defined as
$$
\zeta_k = \mathbb{E}[T_{i,M_i}\mid Z_{i,k}=1].
$$
These three statistics are computed from the model parameters. In particular, from the parameter $\bGamma_k$ of the finite Markov chain, we define $\ba_k$ as the vector whose $e$-th element is $\ba_k[e] = \bGamma_k[1,e]$ for $e \in {1, \ldots, E}$, corresponding to the probability that $S_{i,1} = e$ given $Z_{i,k} = 1$. We also set $r_k = \bGamma_k[1,E+1]$ as the probability that $S_{i,1} = E+1$ given $Z_{i,k} = 1$, $\bQ_k$ as the $E \times E$ sub-stochastic matrix obtained from $\bGamma_k$ by removing its first row and first column, and $\bR_k$ as the $E$-dimensional vector formed by the last column of $\bGamma_k$ excluding its first element.
We denote by $\bF_k$ the $E \times E$ fundamental matrix, defined by
$$
\bF_k=(\bI_E - \bQ_k)^{-1},
$$
where $\bI_E$ denotes the identity matrix of size $E\times E$.

The following lemma provides explicit formulas to compute $\lambda_k$ and $\boldsymbol{\kappa}_k$ from $\bGamma_k$. In particular, it offers a convenient way to summarize the transition matrix $\bGamma_k$ by reporting the expected number of events in a possession, as well as the expected number of occurrences for each event, for a given cluster. This allows for a direct comparison of clusters based on both the total number of events and the event-specific frequencies.

\begin{lemma}\label{lem:expectm}
Under Assumption~\ref{ass:ID}.\ref{ass:Gamma}, the conditional expectation of the number of events in a possession, excluding the event corresponding to the gain of possession, given its component membership, is defined by
$$
\lambda_k =1+\ba_k^\top \bF_k\ones_E, %r_k+\ba_k^\top \bF_k \ones_E ,
$$
where $\ones_{E}$ is the vector of $E$ ones and in addition  
$$
\boldsymbol{\kappa}_k = \ba_k^\top \bF_k,
$$
where $\boldsymbol{\kappa}_k $ is the conditional  expectation of $\bV_{i}$ given $Z_{i,k}=1$, $\bV_{i}$ being the vector whose element $e$ denotes the number of visits of events $e$ for possession $i$.
\end{lemma}

The following lemma provides a formula to compute $\zeta_k$ from $\bGamma_k$ and $\boldsymbol{\mu}_k = (\mu_{k,1}, \ldots, \mu_{k,E+1})$, where $\mu_{k,e}$ is the conditional expectation of $\Delta t_{i,j}$ given $Z_{i,k} = 1$ and $S_{i,j} = e$, defined by $\mu_{k,e} = \rho_{k,e,1} \rho_{k,e,2}$. This allows for a comparison of clusters based on the duration of possessions.
\begin{lemma}\label{lem:expectlength}
Under Assumption~\ref{ass:ID}.\ref{ass:Gamma}, the conditional expectation of the length of a possession that belongs to component $k$, is given by
$$
\zeta_k=     \bmu_k^\top \begin{bmatrix}
    \boldsymbol{\kappa}_k \\
    1
\end{bmatrix}.
$$
\end{lemma}

\subsection{Connection to the approach of \citet{Narayanan}}
Our approach is related to the flexible marked spatio-temporal point process framework of  \citet{Narayanan}, but it operates at a different statistical scale. Whereas their model is fitted to the entire sequence of events in a match thus pooling information over a large number of possessions, we focus on the possession as the basic statistical unit. This change in scale fundamentally alters the dependency structure in the data: at the match level, long-range temporal dependencies and event-type interactions may emerge from the aggregation of many possessions, while at the possession level, sequences are short, with fewer events and more limited temporal depth. As a result, the amount of information available per statistical unit is smaller in our setting, which motivates the use of a simpler specification for the transition law between states compared to the richer parametrization adopted by \citet{Narayanan}

In their work, \citet{Narayanan} emphasize that football event data do not exhibit self-excitation and, on this basis, advocate modeling inter-event times using a gamma process, which is more general than the Poisson process and can capture both over-dispersion and under-dispersion. We adopt the same assumption and model event times within a possession using a gamma process.

Regarding the spatial component, \citet{Narayanan} discretize the pitch into three zones, producing a coarse spatial representation that facilitates inference and interpretation. In contrast, we retain the original continuous coordinates of events, which allows for a finer-grained characterization of spatial dynamics within possessions. Moreover, our spatial sub-model, based on truncated Gaussian transitions in both pitch dimensions, reduces to a two-dimensional truncated Brownian motion when the variance parameters are held constant across event types and time intervals, thereby establishing a direct connection between our formulation and a well-established continuous-space stochastic process.
 
\section{Maximum likelihood inference}\label{sec:MLE}
	We proposed to estimate the model parameters $\btheta$ by maximizing the log-likelihood function of the model defined by 
    $$\ell(\boldsymbol{\theta} ; \mathbf{X})=\sum_{i=1}^n \ln f(\bX_i;\btheta).$$
    This maximization is achieved via an Expectation-Maximization (EM) algorithm \citep{dempster1977maximum} that relies on the completed-data log-likelihood defined by 
	
	$$
	\ell_c(\boldsymbol{\theta}; \mathbf{X}, \mathbf{Z}) = \sum_{i=1}^n \sum_{k=1}^K Z_{ik} \Bigg[  \ln \pi_k  + \ln g(\bS_i;\bGamma_k) + \ln h(\bT_i\mid \bS_i; \boldsymbol{\rho}_k) + \ln \ell(\bU_i\mid\bS_i,\bT_i; \bbeta_k)\Bigg]. 
	$$ 

	The EM algorithm is an iterative algorithm randomly initialized at the model parameter $\btheta^{[0]}$. It alternates between two steps: the expectation step (E-step) which consists of computing the expectation of the complete-data log-likelihood under the current parameters and the maximization step (M-step), which consists of maximizing this expectation with respect to the model parameters. In our situation, its maximization with respect to $(\brho_k,\bbeta_k)$  does not provide a closed-form solution but it can achieved by standard optimization algorithms. To reduce the computational costs, instead of maximizing the expectation of the complete-data log-likelihood, it is sufficient to ensure that it increases at each M-step. This would lead to a Generalized Expectation-Maximization (GEM) algorithm \citep{mclachlan2008algorithm} that keeps the property of increasing the log-likelihood at each iteration. In our context, we consider the GEM such that the GM-step is defined by a maxmization over $\bpi$ and $\bGamma_k$ while independent optimizations of with respect to each $\brho_{k,e}$, and $\eta_{k,e}$ are performed by one iteration of the L-BFGS-B quasi-Newton algorithms initialized at the value of the model parameter obtained at the previous iteration of the GEM algorithm. Recall, that the L-BFGS-B method estimates the Hessian approximation from successive gradients without requiring closed-form second-order derivatives. %This increasing would be ensure by providing only one iteration of the Newton-Raphson algorithm started at the current value of the parameters of the process. The resulting algorithm is therefore a Generalized Expectation-Maximization (GEM; ) algorithm, which reduces the computational complexity of exact maximization while still ensuring convergence (see \citet{rai1993improving}). 
    Iteration $[r]$ of the algorithm is defined by:
\begin{itemize}
		\item E-step: computing the conditional probabilities that sequence $i$ belongs to cluster $k$, given the observed data  and current parameter estimates 
		$$
		r_{i,k}(\btheta^{[r]})=\frac{ \pi_k^{[r]} g(\bS_i;\bGamma_k^{[r]})h(\bT_i\mid \bS_i; \boldsymbol{\rho}_k^{[r]})\ell(\bU_i\mid\bS_i,\bT_i;\bbeta_k^{[r]})}{f(\bX_i;\btheta^{[r]})},
        $$
        then computing the sufficient statistics $n_{k,e}^{(1)}(\btheta^{[r]})=\sum_{i=1}^n \sum_{j=1}^{M_i} r_{i,k}(\btheta^{[r]}) \mathds{1}_{\{S_{i,j}=e\}}$, $n_{k,e}^{(\Delta t)}(\btheta^{[r]})= \sum_{i=1}^n \sum_{j=1}^{M_i}  r_{i,k}(\btheta^{[r]}) \mathds{1}_{\{S_{i,j}= e\}}  \Delta t_{i,j}$ and $n_{k,e}^{(\ln\Delta t)}(\btheta^{[r]})= \sum_{i=1}^n \sum_{j=1}^{M_i}  r_{i,k}(\btheta^{[r]}) \mathds{1}_{\{S_{i,j}= e\}}  \ln \Delta t_{i,j}$.\\
		\item GM-step:
        \begin{itemize}
            \item Maximizing over the proportions and the parameters of the Markov chains
        $$
         \pi_k^{[r+1]}=\dfrac{1}{n}\sum_{i=1}^n r_{i,k}(\btheta^{[r]}) \text{ and }
         \bGamma_k^{[r+1]}[e, \tilde e]=\dfrac{1}{n_{k,e}^{(1)}(\btheta^{[r]})}\sum_{i=1}^n \sum_{j=1}^{M_i} r_{i,k}(\btheta^{[r]}) \mathds{1}_{\{S_{i,j-1}=e, S_{i,j}=\tilde e\}}.
         $$
                  \item Updating the parameters of the Gamma process using the L-BFGS-B quasi-Newton method:
    $$
    \brho_{k,e}^{[r+1]} = \texttt{L-BFGS-B}\left(\brho_{k,e}^{[r]}, \bG_{k,e}(\cdot; \btheta^{[r]})\right),
    $$
    with $\bG_k(\brho_k; \btheta^{[r]})$ denotes the gradient of the expected complete-data log-likelihood with respect to $\brho_{k,e}$ defined by
     $$
\bG_{k,e}(\brho_{k,e}; \btheta^{[r]})=\begin{bmatrix}
n_{k,e}^{(\ln\Delta t)}(\btheta^{[r]}) - n_{k,e}^{(1)}(\btheta^{[r]}) \left[  \ln (\rho_{k,e,2} ) +  \psi(\rho_ {k,e,1}) \right ] \\
  n_{k,e}^{(\Delta t)}(\btheta^{[r]})/\rho_{k,e,2}^2 - n_{k,e}^{(1)}(\btheta^{[r]}) \left( \rho_{k,e,1} / \rho_{k,e,2}\right) 
\end{bmatrix},
$$
where $\psi$ is the digamma function defined by $\psi(u)=\frac{\partial}{\partial u} \ln \Gamma(u)$. 

 \item Updating the parameters of the truncated Gaussian distributions using the L-BFGS-B quasi-Newton method:
    $$
    \eta_{k,h,e}^{[r+1]} = \texttt{L-BFGS-B}\left(\eta_{k,h,e}^{[r]}, \overline\bG_{k,h,e}(\cdot; \btheta^{[r]})\right),
    $$
 with $\overline\bG_{k,h,e}(\cdot; \btheta^{[r]})$ denotes the partial derivative of the expected complete-data log-likelihood with respect to $\eta_{k,h,e}$ defined by
$$ 
\overline\bG_{k,h,e}(\eta_{k,h,e}; \btheta^{[r]})= \sum_{i=1}^n \sum_{j=1}^{M_i} r_{i,k}(\btheta^{[r]}) \mathds{1}_{\{S_{i,j} = e\}} 
\left( \frac{ (\Delta u_{i,j,h})^2}{\eta_{k,h,e}^3} +  \frac{\zeta\left(\frac{\Delta b_{1,i,j,h}}{\eta_{k,h,e}},\frac{\Delta b_{2,i,j,h}}{\eta_{k,h,e}}\right)-1}{\eta_{k,h,e}}
 \right),
$$
where  
$\zeta(a,b) = \frac{b \phi(b) - a\phi(a)}{\Phi(b) - \Phi(a)}$.
\end{itemize}
\end{itemize}

\section{Numerical experiments}\label{sec:simu}
The purpose of this section is to illustrate, using simulated data, the behavior of the proposed clustering method.
Parameter estimation procedure used a multi-start strategy to avoid the risk of convergence to local optima. Specifically, 1000 random initializations were generated, and the 100 best candidates after 10 EM iterations were retained for refinement with 500 EM iterations. Computations were carried out on a high-performance computing cluster using 80 CPU cores in parallel. Specifically, we evaluate the accuracy of the estimated partition as well as the estimated values of the indices $\lambda_k$, $\boldsymbol{\kappa}_k$, and $\zeta_k$, since these indices are used to interpret the estimated clusters. In these simulations, each sample consists of $n$ independent observations generated from the model \eqref{eq:model1}–\eqref{eq:model4} with $K=3$ components in equal proportions and $E=5$. The parameters of the components are specified according to a scalar parameter $\tau$, which allows us to consider different levels of separation between the components. For each component $k$, we consider the following transition matrix,
$$
\Gamma_k[e,\tilde e] = \frac{1}{(1+\exp(k \tau))(E+1)} + \frac{\exp(k \tau)}{(1+\exp(k \tau))} \mathds{1}_{e=\tilde e}, 
$$
the following parameters for the gamma distributions
$$
\rho_{k,e,1}=1+ k \tau \text{ and } \rho_{k,e,2}=1,
$$
and the following parameters for the truncated Gaussian distributions defined on $[0,120]\times[0,80]$
$$
 \eta_{k,h,e} = 1 + k\tau, \quad \text{for } h = 1, 2.
$$

In our experiments, three values of $\tau$ are considered to define the following scenarios: \emph{easy case} (\emph{i.e.}, $\tau=0.65$, corresponding to a classification error of $5\%$), \emph{intermediate case} (\emph{i.e.}, $\tau=0.50$, corresponding to a classification error of $10\%$), and \emph{hard case} (\emph{i.e.}, $\tau=0.40$, corresponding to a classification error of $15\%$). In addition, we consider different sample sizes $n \in \{50, 100, 200, 400\}$. For each scenario and each sample size, 100 independent datasets are generated. To assess the accuracy of the clustering method, maximum likelihood inference is performed on each sample, and each observation is then assigned to the most likely cluster. To compare the estimated and true partitions, we compute the Adjusted Rand Index (ARI; \citet{hubert1985comparing}). Table~\ref{tab:ARI} reports the mean and standard deviation of the ARI obtained from the 100 Monte Carlo replications for the different scenarios and sample sizes. The results indicate that the accuracy of the partition improves with increasing sample size, which can be attributed to the consistency property of the maximum likelihood estimator.

\begin{table}[ht]
\centering
\begin{tabular}{rcccc}
  \hline
  scenario & \multicolumn{4}{c}{$n$} \\
\cline{2-5} & 50 & 100 & 200 & 400 \\
  \hline
easy & 0.642 (0.152) & 0.770 (0.106) & 0.829 (0.045) & 0.833 (0.028) \\ 
  intermediate & 0.444 (0.090) & 0.506 (0.108) & 0.605 (0.109) & 0.698 (0.041) \\ 
  hard & 0.356 (0.097) & 0.378 (0.065) & 0.431 (0.068) & 0.524 (0.073) \\ 
   \hline
\end{tabular}
\caption{Mean and standard deviation in parenthesis of the Adjusted Rand Index computed between the true partition and the estimated partition obtained on 100 Monte Carlo replications for different scenarios and sample sizes.}\label{tab:ARI} 
\end{table}

To evaluate the usefulness of the statistics presented in Section~\ref{subsec:prop}, we compute the estimators of the conditional expectation of the number of events in a possession (\emph{i.e.,} $\lambda_k$), the conditional expectation of the number of visits of each transient events (\emph{i.e.,} $\boldsymbol{\kappa}_k$), and the conditional expectation of the duration of a possession (\emph{i.e.,} $\zeta_k$) based on the maximum likelihood estimate. Table~\ref{tab:index} reports the mean and standard deviation of the Euclidean norm of the difference between each index and its estimator, obtained from the 100 Monte Carlo replications for the different scenarios and sample sizes. The results indicate that the accuracy of the index estimators improves with increasing sample size, which, once again, can be attributed to the consistency property of the maximum likelihood estimator.
\begin{table}[ht]
\centering
\begin{tabular}{rccccc}
  \hline
index &  scenario & \multicolumn{4}{c}{$n$} \\
\cline{3-6} & & 50 & 100 & 200 & 400 \\
  \hline
$\lambda_k$ &easy & 0.256 (0.303) & 0.167 (0.140) & 0.108 (0.072) & 0.070 (0.049) \\ 
 & intermediate & 1.049 (0.868) & 0.684 (0.464) & 0.327 (0.256) & 0.189 (0.104) \\ 
 & hard & 1.397 (0.965) & 1.004 (0.561) & 0.751 (0.481) & 0.383 (0.203) \\ 
   \hline
$\boldsymbol{\kappa}_k$ & easy & 0.167 (0.178) & 0.105 (0.067) & 0.070 (0.030) & 0.045 (0.019) \\ 
&  intermediate & 0.675 (0.426) & 0.409 (0.217) & 0.196 (0.119) & 0.116 (0.042) \\ 
 & hard & 0.878 (0.489) & 0.607 (0.252) & 0.463 (0.224) & 0.225 (0.093) \\  
   \hline
$\zeta_k$ &easy & 8.132 (6.159) & 3.659 (3.369) & 1.762 (0.777) & 1.344 (0.657) \\ 
 & intermediate & 9.103 (3.589) & 7.173 (4.089) & 3.599 (3.108) & 1.491 (0.639) \\ 
 & hard & 6.263 (2.881) & 5.384 (2.050) & 4.654 (2.363) & 2.106 (1.381) \\ 
   \hline
\end{tabular}
\caption{Mean and standard deviation in parenthesis of the  Euclidean norm of the difference between each index and its estimator computed between the true partition and the estimated partition obtained on 100 Monte Carlo replications for different scenarios and sample sizes.}\label{tab:index} 
\end{table}

\section{Analysis of soccer events data} \label{sec:analysis}
In this section we analyze the event data presented in Section~\ref{data::description}.
\subsection{Parameter estimation and model selection}

Model parameters were estimated via maximum likelihood using the Expectation-Maximization (EM) algorithm, under the same experimental conditions as in the simulation study (Section~\ref{sec:simu}). The algorithm was run for a range of component numbers $K \in \{1, \dots, 6\}$. The total computation time was approximately 3.3 hours of user CPU time. The number of components was selected using the Bayesian Information Criterion (BIC), a standard approach in clustering contexts \cite{schwarz1978estimating}. For a model with $K$ components, the BIC is defined as
\begin{equation*}
\mathrm{BIC}_{K} = \ell_{K} (\hat{\btheta}) - \frac{\nu_{K}}{2}\log(n_{\text{tot}}),
\end{equation*}
where $\ell_{K} (\hat{\btheta})$ is the maximized log-likelihood, $\nu_{K}$ is the number of free parameters, and $n_{\text{tot}}$ is the total number of transitions across all possessions. Table~\ref{tab:bic_summary} reports the BIC values for each value of $K$. The model with $K = 4$ components yielded the highest BIC and was therefore selected for further analysis.

\begin{table}[h]
\centering
\begin{tabular}{c c c c c c c}
\toprule
$K$   & 1 & 2 & 3 & 4 & 5 & 6 \\
\hline
BIC   & -269344.3 & -265120.1 & -264383.3 & \textbf{-264121.2} & -264141.4 & -264613.1 \\
\bottomrule
\end{tabular}
\caption{Bayesian Information Criterion (BIC) values for models with different numbers of components $K$. The model with $K = 4$ maximizes the BIC and is selected for subsequent analysis.}
\label{tab:bic_summary}
\end{table}

\subsection{Characteristics of possession clusters}

Our mixture model partitions possessions into four distinct clusters, each defined by its relative proportion in the dataset and its spatio-temporal characteristics. Cluster~1 accounts for 37\% of possessions, Cluster~2 for 31\%, Cluster~3 for 13\%, and Cluster~4 for 19\%. This distribution indicates a dominance of short and medium-length possessions (Clusters~1 and~2), while elaborated and long sequences (Clusters~3 and~4) remain less frequent. Each cluster represents a distinct possession profile: 

\begin{itemize}
    \item \textbf{Cluster~1:} short, direct sequences ($\lambda_1 \approx 5.6$, $\zeta_1 \approx 25.9$~sec.), dominated by rapid passing, limited ball carrying, and moderate spatial spread typical of fast transitional play.
    \item \textbf{Cluster~2:} medium-length, active possessions ($\lambda_2 \approx 11.6$, $\zeta_2 \approx 28.4$~sec.), involving a balanced mix of passes, carries, and duels, often emerging in dynamic transition phases.
    \item \textbf{Cluster~3:} intermediate possessions ($\lambda_3 \approx 11.0$, $\zeta_3 \approx 19.7$~sec.), marked by frequent carrying actions and compact spatial structures, suggesting intense, contested play under pressure.
    \item \textbf{Cluster~4:} the most elaborate possessions ($\lambda_4 \approx 15.0$, $\zeta_4 \approx 33.1$~sec.), characterized by extended passing and carrying sequences with broad pitch coverage, indicative of patient build-up and sustained control.
\end{itemize}

To further characterise the clusters beyond their transition dynamics, we use the three statistical indicators presented in Lemmas~\ref{lem:expectm} and~\ref{lem:expectlength}: the conditional expectation of the number of events in a possession ($\lambda_k$), the conditional expectation of the number of visits of each transient event ($\boldsymbol{\kappa}_k$), and the conditional expectation of the duration of a possession ($\zeta_k$), all estimated via maximum likelihood. Table~\ref{tab:possession_summary} reports these statistics across the four clusters. The clustering reveals two main dimensions of variation: structural complexity, captured by $\lambda_k$ and $\kappa_k$, and temporal intensity, measured by $\zeta_k$. These dimensions distinguish possession strategies that range from short, direct transitions to longer, more elaborate build-up sequences.

\begin{table}[ht]
\centering
\begingroup
\begin{tabular}{lrrrr}
\toprule
 & Cluster 1 & Cluster 2 & Cluster 3 & Cluster 4 \\
\midrule
$\pi_k$ (proportion) & 0.37 & 0.31 & 0.13 & 0.19 \\
$\lambda_k$ (events) & 5.60 & 11.56 & 11.01 & 15.04 \\
$\zeta_k$ (seconds) & 25.88 & 28.41 & 19.65 & 33.13 \\
$\kappa_k^{\text{Pass}}$ & 2.21 & 4.89 & 4.05 & 6.62 \\
$\kappa_k^{\text{Carries}}$ & 0.98 & 3.71 & 3.85 & 5.01 \\
$\kappa_k^{\text{Pressure}}$ & 0.79 & 1.18 & 1.58 & 1.46 \\
$\kappa_k^{\text{Ball Receipt}}$ & 0.21 & 0.23 & 0.08 & 0.25 \\
$\kappa_k^{\text{Clearance}}$ & 0.13 & 0.17 & 0.08 & 0.14 \\
$\kappa_k^{\text{Block}}$ & 0.06 & 0.07 & 0.06 & 0.09 \\
$\kappa_k^{\text{Shot}}$ & 0.05 & 0.05 & 0.10 & 0.09 \\
$\kappa_k^{\text{Duel}}$ & 0.06 & 0.09 & 0.07 & 0.07 \\
$\kappa_k^{\text{Miscontrol}}$ & 0.06 & 0.06 & 0.06 & 0.11 \\
$\kappa_k^{\text{Ball Recovery}}$ & 0.03 & 0.06 & 0.03 & 0.06 \\
$\kappa_k^{\text{Goal Keeper}}$ & 0.01 & 0.02 & 0.01 & 0.04 \\
$\kappa_k^{\text{Interception}}$ & 0.02 & 0.04 & 0.03 & 0.07 \\
$\kappa_k^{\text{Foul Won}}$ & 0.00 & 0.01 & 0.00 & 0.03 \\
\bottomrule
\end{tabular}
\endgroup
\caption{Summary of the four possession clusters, including their relative proportions ($\pi_k$), expected number of events per possession ($\lambda_k$), expected duration in seconds ($\zeta_k$), and the conditional expectations of each event type ($\kappa_k^{\text{event}}$).}
\label{tab:possession_summary}
\end{table}
\subsection{Tactical interpretation and match context}

Our clustering results align with established tactical typologies in football analytics. Cluster 1 (short, direct possessions) resembles \textit{counter-attacking} patterns described by \citealp{fernandez2024reporting}, characterized by rapid vertical progression with minimal ball circulation. Cluster 4 (elaborate build-up) corresponds to \textit{positional play} or \textit{tiki-taka} styles, emphasizing patient possession and spatial control. The emergence of intermediate Clusters (2 and 3) suggests a continuum of tactical approaches rather than discrete categories, consistent with recent findings on possession diversity.

Clusters 1 and 2 correspond to short and direct possessions. These sequences show inter-arrival time distributions concentrated near zero, reflecting rapid ball circulation and frequent events within compact spatio-temporal windows. In contrast, Cluster~4 captures longer and more elaborate possessions, with heavier-tailed temporal distributions and broader spatial dispersions, indicative of gradual build-up and wide pitch coverage. Cluster~3 occupies an intermediate position, balancing speed and structural density. Spatially, forward-oriented actions such as passes and carries display asymmetric spreads along the horizontal axis, while duels and interceptions remain tightly concentrated vertically, indicating vertical compactness. The joint modelling of temporal and spatial distributions confirms the heterogeneity of match dynamics and distinct possession styles.

At the team and match level, the prevalence of each cluster varies with tactical context and venue. Cluster~1 often dominates away matches, favoring quick transitions and reactive play. For instance, Nîmes on matchday~3 (58.6\%) and Strasbourg on matchday~12 (57.6\%) relied heavily on Cluster~1, achieving wins or draws through efficient transitions and exploitation of opponent space. Similarly, Brest (58.8\%, matchday~9) and Marseille (54.1\%, matchday~15) successfully translated strong Cluster~1 profiles into victories.

Cluster~2, characterized by dense and short possessions, typically emerges under high pressing or compact mid-blocks. Examples include Angers (50.8\%, matchday~8) and Saint-Étienne (46.8\%, matchday~5), with mixed outcomes, highlighting that the success of Cluster~2 depends on execution quality and territorial control. Strong Cluster~2 patterns during defeats (e.g., Lorient, 50.7\%, matchday~23; Lyon, 39.7\%, matchday~28) suggest that this profile can reflect reactive containment rather than sustained offensive progression.

Cluster~3 represents intermediate or transitional possessions, often appearing during adaptation or instability in buildup phases. Peaks such as Metz (16.9\%, matchday~30) and Bordeaux (19.4\%, matchday~35) occur in matches with mixed outcomes, reflecting tactical inconsistency or challenges in maintaining rhythm under pressure. 

Cluster~4 reflects the most elaborate possessions, with extended, spatially expansive sequences. Teams such as Reims (22.4\%, matchday~6; 25.4\%, matchday~31), Angers (35.6\%, matchday~8), and Dijon (26.9\%, matchday~34) use Cluster~4 to control the game and impose structure. However, outcomes vary Reims draws, Dijon wins, Angers loses showing that possession elaboration alone does not ensure success and must be coupled with effective offensive penetration.

Overall, a clear pattern emerges between cluster dominance and match result. Wins are often associated with Cluster~1 or Cluster~4, while defeats occur more frequently with Cluster~2 dominance. Draws exhibit balanced cluster distributions, reflecting equilibrium between control and adaptation. Away teams favor Cluster~1, emphasizing direct transitions, whereas home teams engage more in Clusters~2 and~4, aligning with positional play and initiative-taking. Tracking these latent clusters over time allows coaches and analysts to identify when tactical styles translate into competitive advantage and how teams alternate between control, reactivity, and adaptability. Figure~\ref{fig:boxplot_match_outcome} illustrates systematic patterns of cluster usage across match outcomes, confirming the interpretations above. Finally, Table~\ref{tab:clusters_wide} in Appendix reports the distribution of clusters across the 38 matches considered. 

\subsection{Practical implications for training}
The identified clusters enable targeted training scenario design for Rennes goalkeepers. For instance, sessions focusing on \textit{counter-attack defense} (Cluster 1) can be populated exclusively with possessions from this cluster, ensuring realistic temporal dynamics (mean duration $\approx 26$ sec.) and spatial patterns (shots occurring at mean distance 18.2 $m$ from goal, as inferred from the spatial model). Moreover, tracking cluster distributions over time (Table \ref{tab:clusters_wide}) allows coaches to identify opponents' tactical adaptations mid-season. For example, the shift in Marseille's cluster profile between matchdays 15 (54\% Cluster 1) and 22 (43\% Cluster 2) suggests a strategic adjustment that could inform pre-match preparation. Future integration with virtual reality (VR) platforms will leverage our generative model to simulate novel but plausible possession sequences, augmenting real data with synthetic scenarios.

\begin{figure}[htp!]
    \centering
    \includegraphics[width=0.8\linewidth]{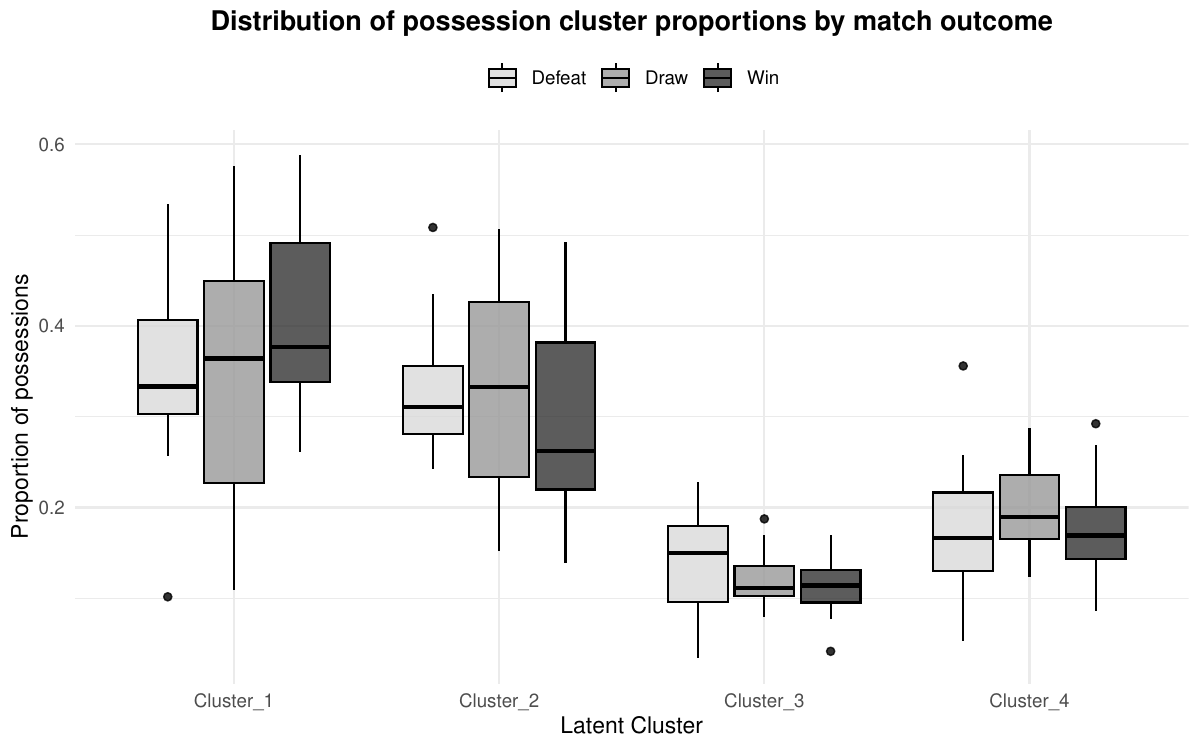}
    \caption{Distribution of possession cluster proportions according to match outcomes (win, draw, defeat). Winning teams tend to exhibit higher proportions of Cluster 1 (direct and reactive play) and Cluster 4 (elaborated build-ups), while Cluster 2 (dense and short possessions) is more prominent in defeats. Cluster 3 remains generally low and less sensitive to outcome. This suggests a performance advantage for either direct transitional or structured positional play, depending on context.}
    \label{fig:boxplot_match_outcome}
\end{figure}

\section{Conclusion}\label{sec:concl}
In this work, we introduced a novel mixture model for clustering marked spatio-temporal point processes, in which clusters are characterized jointly by the distribution of event types (marks), their arrival times, and their spatial locations. Methodologically, the main contribution lies in integrating the full joint distribution of marks, inter-event times, and spatial coordinates within a probabilistic mixture framework, thus enabling clustering at the level of entire point process realizations. This allows the identification of latent patterns that emerge only when considering the complete temporal and spatial structure of sequences, rather than treating events in isolation. Applied to football analytics, our framework makes it possible to cluster entire possessions sequences of on-ball actions into distinct tactical scenarios. While our primary motivation was to populate the environment of a virtual reality training tool for Rennes goalkeepers with realistic and diverse possession sequences, the methodology also offers a powerful tool for performance analysis. In particular, it can be used to assess whether a team’s style of play or defensive organization exhibits significant shifts during a match, providing both practitioners and analysts with actionable insights.

Beyond football, our modeling framework can be adapted to event data from a wide range of sports, provided that their specific characteristics are taken into account. For example, in rugby, where passes must be made backward, the conditional distribution of spatial coordinate differences between events would need to be reformulated, as a truncated normal distribution may not be appropriate. Similarly, in basketball, the 24-second shot clock imposes a hard temporal limit on possessions, which must be explicitly incorporated into the arrival-time distribution of events, with shots treated as terminal events. Nevertheless, these adaptations are straightforward thanks to the model’s conditional structure on the current event type, making it flexible and readily transferable to other sporting contexts.

In sum, our approach bridges the gap between advanced point process modeling and practical sports analytics, offering a versatile, interpretable, and domain-adaptable tool for uncovering and exploiting the hidden structures in event-driven data.

\begin{funding}
This research was funded by the French Agence Nationale de la Recherche, grant number ANR-19-STPH-004 and conducted within the framework of the PIA EUR DIGISPORT project (ANR-18-EURE-0022). 
\end{funding}

\begin{Conflicts of interest}
  The authors declare that they have no conflicts of interest.  
\end{Conflicts of interest}

%\bibliographystyle{plain}
%\bibliography{reference}

%USE THE BELOW OPTIONS IN CASE YOU NEED AUTHOR YEAR FORMAT.
%\bibliographystyle{abbrvnat}
%\bibliography{reference}

\begin{appendices} 
\section{Proofs of the mathematical results} \label{sec:proofs}
\begin{proof}[Proof of Lemma~\ref{lem:idengrl}]
The parameters of \eqref{eq:model1}-\eqref{eq:model4} are identifiable up to label swapping if
\begin{equation} \label{eq:startID}
\forall \bx,\, f(\bx ; \btheta) = f(\bx ; \widetilde\btheta),
\end{equation}
implies
$$ \btheta \equiv \widetilde\btheta,$$
where $\btheta  \equiv \widetilde \btheta $ means that it exists a permutation of the labels of the $K$ components of $\widetilde\btheta$, denoted by $\sigma$, such that $\btheta=\sigma(\widetilde\btheta)$. 
Without loss of generality, we assume that the labeling of the mixture components induces an ordering among the gamma distributions associated with the fixed event $e_0$. Since the parameters of the gamma distributions are distinct (see Assumption~\ref{ass:ID}.\ref{ass:rho}), we consider that
\begin{equation}\label{eq:ord1}
\rho_{k,e_0,2} \geq \rho_{k+1,e_0,2} \text{ and } \rho_{k,e_0,1} > \rho_{k+1,e_0,1} \text{ if } \rho_{k,e_0,2} =  \rho_{k+1,e_0,2}
\end{equation}
and
\begin{equation}\label{eq:ord2}
\widetilde\rho_{k,e_0,2} \geq \widetilde\rho_{k+1,e_0,2} \text{ and } \widetilde\rho_{k,e_0,1} > \widetilde\rho_{k+1,e_0,1} \text{ if } \widetilde\rho_{k,e_0,2} =  \widetilde\rho_{k+1,e_0,2}.
\end{equation}
In addition, without loss of generality, we consider 
\begin{equation}\label{eq:ord3}
\rho_{k,e_0,2} \geq \widetilde \rho_{k,e_0,2}  \text{ and }  \rho_{1,e_0,1} \geq \widetilde\rho_{1,e_0,1} \text{ if } \rho_{1,e_0,2} =  \widetilde\rho_{1,e_0,2}.
\end{equation}
Consider a fixed sequence $\bs^{(e_0)}=(s^{(e_0)}_0,\ldots,s^{(e_0)}_{M_0})^\top$ of arbitrary length $M_0$ that never takes the value $E+1$ and such that $s^{(e_0)}_{M_0}=e_0$. Note that there exists at least one $M_0$ such that the probability to observe $\bs^{(e_0)}$ is strictly positive under the finite Markov chain defined by $\bGamma_1$ (i.e., $g(\bs^{(e_0)};\bGamma_1)>0$). The existence of such $\bs^{(e_0)}$ is ensured by Assumption~\ref{ass:ID}.\ref{ass:Gamma}. Hence, considering the marginal distribution over the first $M_0$ events and times of arrivals, then \eqref{eq:startID} implies that
\begin{multline}\label{eq:equal}
\forall (\bs^{(e_0)},\Delta\bt, \Delta t_{M_0}),\, \sum_{k=1}^K \omega_k(\bs^{(e_0)},\Delta \bt^{(e_0)};\btheta) \gamma(\Delta t_{M_0};\rho_{k,e_0,1},\rho_{k,e_0,2})\\=\sum_{k=1}^K \omega_k(\bs^{(e_0)},\Delta \bt^{(e_0)};\widetilde\btheta) \gamma(\Delta t_{M_0};\widetilde\rho_{k,e_0,1},\widetilde\rho_{k,e_0,2}),    
\end{multline}
where $\Delta \bt = (\Delta t_1,\ldots,\Delta t_{M_0-1})$ groups some arbitrary $M_0-1$ times elapsed between events and
$$
\omega_k(\bs^{(e_0)},\Delta \bt;\btheta)= \pi_k g(\bs^{(e_0)};\bGamma_k)\left(\prod_{j=1}^{M_{0}-1} \prod_{e=1}^{E+1} \left[\gamma(\Delta t_{j}; \rho_{k,e,1}, \rho_{k,e,2})\right]^{\mathds{1}_{\{s_{j}^{(e_0)}=e\}}}\right).
$$
Since all the proportions are strictly positive (see Assumption~\ref{ass:ID}.\ref{ass:prop}), $\omega_1(\bs^{(e_0)},\Delta \bt^{(e_0)};\btheta)$ is strictly positive. Hence, dividing both sides of \eqref{eq:equal} by $\omega_1(\bs^{(e_0)},\Delta \bt^{(e_0)};\btheta) \gamma(\Delta t_{M_0};\rho_{1,e,1},\rho_{1,e,2})$ leads to
\begin{multline*}
\forall (\bs^{(e_0)},\Delta\bt, \Delta t_{M_0}),\,
1 + \sum_{k=2}^K \frac{\omega_k(\bs^{(e_0)},\Delta \bt^{(e_0)};\btheta) \gamma(\Delta t_{M_0};\rho_{k,e,1},\rho_{k,e,2})}{\omega_1(\bs^{(e_0)},\Delta \bt^{(e_0)};\btheta) \gamma(\Delta t_{M_0};\rho_{1,e_0,1},\rho_{1,e_0,2})} \\ =\sum_{k=1}^K \frac{\omega_k(\bs^{(e_0)},\Delta \bt^{(e_0)};\widetilde\btheta) \gamma(\Delta t_{M_0};\widetilde\rho_{k,e_0,1},\widetilde\rho_{k,e_0,2})}{\omega_1(\bs^{(e_0)},\Delta \bt^{(e_0)};\btheta) \gamma(\Delta t_{M_0};\rho_{1,e,1},\rho_{1,e,2})}.
\end{multline*}
Using the ordering conditions \eqref{eq:ord1}-\eqref{eq:ord3}, then taking the limit as $\Delta t_{M_0} \to \infty$ in the previous equation implies that
$$
\forall (\bs^{(e_0)},\Delta\bt),\, 1 =  \frac{\omega_1(\bs^{(e_0)},\Delta \bt^{(e_0)};\widetilde\btheta) }{\omega_1(\bs^{(e_0)},\Delta \bt^{(e_0)};\btheta) } \lim_{\Delta t_{M_0} \to \infty}  \frac{ \gamma(\Delta t_{M_0};\widetilde\rho_{1,e_0,1},\widetilde\rho_{1,e_0,2})}{ \gamma(\Delta t_{M_0};\rho_{1,e_0,1},\rho_{1,e_0,2})}.
$$
By respecting \eqref{eq:ord3},  $ \lim_{\Delta t_{M_0} \to \infty} \frac{ \gamma(\Delta t_{M_0};\widetilde\rho_{1,e_0,1},\widetilde\rho_{1,e_0,2})}{ \gamma(\Delta t_{M_0};\rho_{1,e_0,1},\rho_{1,e_0,2})}$ is equal to one only if $\brho_{1,e_0}=\widetilde\brho_{1,e_0}$ and it is equal to zero otherwise. Thus, the previous equation implies that
$$
\brho_{1,e_0}=\widetilde\brho_{1,e_0} \text{ and } \forall(\bs^{(e_0)},\bt^{(e_0)}),\, \omega_1(\bs^{(e_0)},\Delta \bt^{(e_0)};\btheta)=\omega_1(\bs^{(e_0)},\Delta \bt^{(e_0)};\widetilde\btheta).
$$
Repeating the same reasoning over the indices of the components leads to
$$
\forall k\in\{1,\ldots,K\}, \brho_{k,e_0}=\widetilde\brho_{k,e_0} \text{ and } \forall(\bs^{(e_0)},\bt^{(e_0)}),\, \omega_k(\bs^{(e_0)},\Delta \bt^{(e_0)};\btheta)=\omega_k(\bs^{(e_0)},\Delta \bt^{(e_0)};\widetilde\btheta).
$$
Since Assumption~\ref{ass:ID}.\ref{ass:Gamma} ensures that any event in $\{1,\ldots,E+1\}$ can be reached from the initial event 0, we can repeat the same reasoning over the indices of the events (note that this would require a new ordering between components  \eqref{eq:ord1}-\eqref{eq:ord3}). This leads in particular to
$$
\brho=\widetilde\brho.
$$
Using the previous equality and considering the marginal distribution of $(\bS,\bT)$, then \eqref{eq:startID} implies that
$$
\forall(\bs,\bt),\, \sum_{k=1}^K \left( \pi_k g(\bs;\bGamma_k) - \widetilde\pi_k \widetilde g(\bs;\bGamma_k)\right) h(\bt \mid \bs;\brho_k) = 0.
$$
Since $\brho_k \neq \brho_{\ell}$ for $k\neq \ell$, then we have that the family of density functions $\{h(\bt \mid \bs;\brho_1),\ldots,h(\bt \mid \bs;\brho_K)\}$ is composed of $K$ linearly independent functions. Therefore, we have
$$
\forall \bs,\, \forall k,\,   \pi_k g(\bs;\bGamma_k) = \widetilde\pi_k \widetilde g(\bs;\bGamma_k).
$$ 
Summing over all the possible values of $\bs$ implies that $\widetilde \pi_k = \pi_k$, then considering all the specific couples of events implies that $\widetilde \bGamma_k=\bGamma_k$. Using the previous equality, then \eqref{eq:startID} implies that
$$
\forall (\bs,\bt,\bu),\, \sum_{k=1}^{K} \pi_k g(\bs;\bGamma_k)h(\bt\mid \bs; \boldsymbol{\rho}_k)(\ell(\bu\mid\bs,\bt; \bbeta_k)-\ell(\bu\mid\bs,\bt; \widetilde\bbeta_k))=0. 
$$
Again, using the fact that the family of functions ${ \pi_1 g(\bs;\bGamma_1)h(\bt \mid \bs;\brho_1),\ldots, \pi_K g(\bs;\bGamma_K)h(\bt \mid \bs;\brho_K)}$ is composed of $K$ linearly independent functions, we obtain that $\widetilde\bbeta=\bbeta$ and thus $\widetilde\btheta=\btheta$, which concludes the proof.
\end{proof}

\begin{proof}[Proof of Lemma~\ref{lem:expectm}]
%Recall that $\bGamma_k[e,\tilde e] = \mathbb{P}(S_{i,1}=\tilde e \mid S_{i,j-1}=e+1 \mid Z_{ik}=1)$ for any integer $j$ greater than one, $\mathbb{P}(S_{i,0}=1\mid Z_{ik})=1)=1$, $\mathbb{P}(S_{i,j}=1\mid Z_{ik})=1)=0$ for any $j>0$ and that event $E+1$ is an absorbing state that implies the death of the process. 
To compute the conditional expectation of the number of events of a possession given its cluster membership (\emph{i.e.,} the conditional expectation of $M_i$ given $Z_{ik}=1$), we use the following decomposition of the transition matrix
$$
\bGamma_ k  = \begin{bmatrix}
\ba^\top_k &  r_k \\
\bQ_k & \bR_k
\end{bmatrix}.
$$
Noting that $S_{i,0}=0$ with probability one, we have
$$
\mathbb{E}[M_i\mid Z_{i,k}=1] = \sum_{e=1}^{E+1}\mathbb{E}[M_i\mid Z_{i,k}=1,S_{i,1}=e] \mathbb{P}(S_{i,1}=e\mid S_{i,0}=0,Z_{i,k}=1).
$$
Note that $M_i$ is equal to the first index $j$ such that $S_{i,j}=E+1$. Since $E+1$ is an absorbing state, we have
$$
\forall e\in\mathcal{E},\, \mathbb{E}[M_i\mid Z_{i,k}=1,S_{i,1}=e] =\sum_{j=1}^\infty j \mathbb{P}(S_{i,j}= E+1, \{S_{i,j-1}\neq E+1\} \mid S_{i,1}=e, Z_{i,k}=1). 
$$
Noting that  $M_i=1$  when $S_{i,1}=E+1$ and that $\mathbb{P}(S_{i,1}=E+1\mid S_{i,0}=0)=r_k$, we have
$$
\mathbb{E}[M_i\mid Z_{i,k}=1] = r_k + \sum_{e=1}^{E}\mathbb{E}[M_i\mid Z_{i,k}=1,S_{i,1}=e]\mathbb{P}(S_{i,1}=e\mid S_{i,0}=0,Z_{i,k}=1).
$$
The conditional probability to reach the absorbing state at the event $j$ (\emph{i.e.,} $S_{i,j}=E+1$ and $S_{i,j-1}\neq E+1$) conditionally on $Z_{i,k}=1$ and $S_{i,1}=e$ is defined for any $e\in\{1,\ldots,E\}$ by
\begin{multline*}
\mathbb{P}(S_{i,j}= E+1, \{S_{i,j-1}\neq E+1\} \mid S_{i,1}=e, Z_{i,k}=1) \\= \sum_{\tilde e =1}^E \mathbb{P}(S_{i,j-1}= \tilde e \mid S_{i,1}=e, Z_{i,k}=1)\mathbb{P}(S_{i,j}= E+1 \mid S_{i,j-1}=\tilde e,S_{i,1}=e, Z_{i,k}=1).
\end{multline*}
If $j=1$ then $\mathbb{P}(S_{i,j}= E+1, \{S_{i,j-1}\neq E+1\} \mid S_{i,1}=e, Z_{i,k}=1)=0$ since $e\neq E+1$. If $j=2$, $\mathbb{P}(S_{i,j}= E+1, \{S_{i,j-1}\neq E+1\} \mid S_{i,1}=e, Z_{i,k}=1)$ is equal to the element $e$ of vector $\bR_k$. In addition, we have $\mathbb{P}(S_{i,j}= E+1 \mid S_{i,j-1}=\tilde e,S_{i,1}=e, Z_{i,k}=1)=\mathbb{P}(S_{i,j}= E+1 \mid S_{i,j-1}=\tilde e, Z_{i,k}=1)$ and this probability corresponds to the element $e$ of vector $\bR_k$. Finally, if $j>2$, $\mathbb{P}(S_{i,j-1}= \tilde e \mid S_{i,1}=e, Z_{i,k}=1)$ corresponds to the element $(e,\tilde e)$ of $\bQ_k^{j-1}$. Hence, noting that $\bQ^{0}=\bI_E$, we have
$$
\mathbb{P}(S_{i,j}= E+1, \{S_{i,j-1}\neq E+1\} \mid S_{i,1}=e, Z_{i,k}=1)=\left(\bQ_k^{j-2}\bR_k\right)[e].
$$
For any $e\in\{1,\ldots,E\}$, we have
$$
\mathbb{E}[M_i\mid Z_{i,k}=1, S_{i,1}=e] 
=\sum_{j=1}^\infty (j +1)\left(\bQ_k^{j-1}\bR_k\right)[e].%\mathbb{P}(S_{i,j}= E+1, \{S_{i,j-1}\neq E+1\} \mid S_{i,1}=e, Z_{i,k}=1),
$$
This implies that
$$
\mathbb{E}[M_i\mid Z_{i,k}=1, S_{i,1}=e] = \sum_{j=0}^\infty  \left(\bQ_k^{j}\bR_k\right)[e] + \sum_{j=1}^\infty j \left(\bQ_k^{j-1}\bR_k\right)[e].    
$$
Since the spectral radius of $\bQ_k$ satisfies $\rho(\bQ_k) < 1$, the Neumann series $\sum_{j=0}^\infty \bQ_k^j$ converges and equals $\bF_k$. Consider the series $\sum_{j=1}^\infty j \bQ_k^{j-1}$. Multiplying it on the left by $(\bI - \bQ_k)$ and regrouping the terms according to the power of $\bQ_k$ shows that the coefficient of $\bQ_k^0$ is $1$ and, for each $j \geq 1$, the coefficient of $\bQ_k^j$ is also $1$. Therefore,
$(\bI-\bQ_k)\sum_{j=1}^\infty \bQ_k^{j-1}=\sum_{j=0}^\infty \bQ_k^j$, 
this implies that $\sum_{j=1}^\infty j \bQ_k^{j-1}=\bF_k^2$ where $\bF_k=(\bI_E - \bQ_k)^{-1}$, leading that
$
\sum_{j=1}^\infty j \left(\bQ_k^{j-1}\bR_k\right) =\bF_k^2\bR_k$ and thus
$$
\mathbb{E}[M_i\mid Z_{i,k}=1, S_{i,1}=e] = \left(\bF_k \bR_k\right)[e] +   \left(\bF_k^2\bR_k\right)[e].
$$
Using the fact that the transition matrix is stochastic, we have 
$\bQ_k \ones_E + \bR_k = \ones_E$, where $\ones_{E}$ is the vector of $E$ ones,  leading that 
 $  \bR_k = \ones_E - \bQ_k \ones_E$ and thus, we have $(\bI-\bQ_k)^{-1}\bR_k=\ones_E$.  Applying the inverse matrix $\bF_k$ once more to both sides yields
$\bF_k^2\bR_k=\bF_k\ones_E$, and thus
$$
\mathbb{E}[M_i\mid Z_{i,k}=1, S_{i,1}=e] =1 + \left(\bF_k\ones_E\right)[e].
$$
Therefore, we have
$$
\mathbb{E}[M_i\mid Z_{i,k}=1] =1+\ba_k^\top \bF_k\ones_E .
$$
The second assertion is obtained by noting the total number of visits to event $e$, given $Z_{i,k}=1$, can be written as the sum over all $j$ of the indicator that $S_{i,j}=e$. Indeed, we have
\begin{align*}
    V_{i,e}&=\sum_{j=1}^{M_i} \mathds{1}_{S_{i,j}=e}\\
 &=\sum_{j=1}^{\infty} \mathds{1}_{S_{i,j}=e}\mathds{1}_{j\leq M_i}.
\end{align*}
Noting that for any $e\in\{1,\ldots,E\}$, $\{S_{i,j}=e\}$ implies that $\{j<M_i\}$ then $\mathds{1}_{S_{i,j}=e}\mathds{1}_{j\leq M_i}=\mathds{1}_{S_{i,j}=e}$ leading that for any $e\in\{1,\ldots,E\}$
$$
 V_{i,e}=\sum_{j=1}^{\infty} \mathds{1}_{S_{i,j}=e}.
$$
Hence, for any $e\in\{1,\ldots,E\}$, the conditional expectation of $V_{i,e}$ given $Z_{i,k}=1$ is simply the sum of these probabilities over $j$.
\begin{align*}
\mathbb{E}[V_{i,e}\mid Z_{i,k}=1] &= \sum_{j=1}^{\infty}\mathbb{E}[ \mathds{1}_{S_{i,j}=e}\mid Z_{i,k}=1] \\
&= \sum_{j=1}^{\infty}\mathbb{P}( S_{i,j}=e\mid Z_{i,k}=1).
\end{align*}
Recall that, for any $e\in\{1,\ldots,E\}$ and $j\geq 1$, we have $\mathbb{P}(S_{i,j}=e \mid Z_{i,k}=1) = \left[\ba_k^\top \bQ_k^{j-1}\right][e]$.
The proof is complete by taking the sum of both sides of the previous equation over all $j$ from 1 to infinity and by noting that $\sum_{j=0}^\infty \bQ_k^j=\bF_k$.
\end{proof}

\begin{proof}[Proof of Lemma~\ref{lem:expectlength}]
Note that 
\begin{align*}
    T_{i,M_i}&=\sum_{j=1}^{M_i} \Delta T_{i,j}\\
    &=\sum_{j=1}^\infty \Delta T_{i,j}\mathds{1}_{j\leq M_i}\\
    &=\Delta T_{i,M_i} + \sum_{j=1}^\infty \Delta T_{i,j}\mathds{1}_{j< M_i}\\
    &=\Delta T_{i,M_i}+\sum_{j=1}^\infty \Delta T_{i,j} \mathds{1}_{S_{i,j}\neq E+1}.
\end{align*}
Since by Assumption~\ref{ass:ID}.\ref{ass:Gamma}, we have that $\sum_{j=1}^\infty \mathbb{P}(S_{i,j}= E+1 \mid Z_{i,k}=1)=1$, and that $\mathbb{E}[ \Delta T_{i,M_i} \mid S_{i,j}=E+1, Z_{i,k}=1]=\mu_{k,E+1}$, then
$$
\mathbb{E}[\Delta T_{i,M_i}\mid Z_{i,k}=1] = \mu_{k,E+1}.
$$
Note that 
$$
\sum_{j=1}^\infty \Delta T_{i,j} \mathds{1}_{S_{i,j}\neq E+1}=\sum_{e=1}^E\sum_{j=1}^\infty \Delta T_{i,j} \mathds{1}_{S_{i,j}=e}.
$$
Since $\sum_{j=1}^\infty \mathbb{P}(S_{i,j}= e \mid Z_{i,k}=1)=\kappa_{k,e}$ where $\kappa_{k,e}$ is the element $e$ of $\boldsymbol{\kappa}_k$ and that  $\mathbb{E}[\Delta T_{i,M_i}\mid S_{i,j}=e, Z_{i,k}=1]=\mu_{k,e}$, then
$$
\mathbb{E}[\sum_{j=1}^\infty \Delta T_{i,j} \mathds{1}_{S_{i,j}=e}\mid Z_{i,k}=1] =\kappa_{k,e}\mu_{k,e}.
$$
There, using the linearity of the expectation provides that
$$
\mathbb{E}[ T_{i,M_i}\mid Z_{i,k}=1] = \bmu_k^\top \begin{bmatrix}
    \boldsymbol{\kappa}_k \\
    1
\end{bmatrix}.
$$
 
\end{proof}

\newpage

\section{Details of the estimated possession model}
\begin{table}[h]
\centering
\resizebox{\textwidth}{!}{%
\begin{tabular}{lcccccccccc}
\hline
\textbf{Team} & \textbf{Matchday} & \textbf{Cluster 1} & \textbf{Cluster 2} & \textbf{Cluster 3} & \textbf{Cluster 4} & \textbf{Location} & \textbf{Score} & \textbf{Status}   \\
\hline
              Lille &         1 &      0.443 &      0.320 &      0.113 &      0.124 & Away & 1-1 &    Draw \\
        Montpellier &         2 &      0.507 &      0.240 &      0.147 &      0.107 & Home & 2-1 &     Win \\
              Nîmes &         3 &      0.586 &      0.172 &      0.155 &      0.086 & Away & 4-2 &     Win \\
          AS Monaco &         4 &      0.377 &      0.262 &      0.098 &      0.262 & Home & 2-1 &     Win \\
      Saint-Étienne &         5 &      0.278 &      0.468 &      0.114 &      0.139 & Away & 3-0 &     Win \\
     Stade de Reims &         6 &      0.155 &      0.483 &      0.138 &      0.224 & Home & 2-2 &    Draw \\
              Dijon &         7 &      0.452 &      0.210 &      0.129 &      0.210 & Away & 1-1 &    Draw \\
             Angers &         8 &      0.102 &      0.508 &      0.034 &      0.356 & Home & 1-2 &  Defeat \\
     Stade Brestois &         9 &      0.588 &      0.216 &      0.078 &      0.118 & Home & 2-1 &     Win \\
Paris Saint-Germain &        10 &      0.407 &      0.356 &      0.085 &      0.153 & Away & 0-3 &  Defeat \\
           Bordeaux &        11 &      0.299 &      0.287 &      0.172 &      0.241 & Home & 0-1 &  Defeat \\
         Strasbourg &        12 &      0.576 &      0.153 &      0.102 &      0.169 & Away & 1-1 &    Draw \\
               Lens &        13 &      0.338 &      0.311 &      0.162 &      0.189 & Home & 0-2 &  Defeat \\
           OGC Nice &        14 &      0.325 &      0.416 &      0.104 &      0.156 & Away & 1-0 &     Win \\
          Marseille &        15 &      0.541 &      0.180 &      0.131 &      0.148 & Home & 2-1 &     Win \\
            Lorient &        16 &      0.439 &      0.281 &      0.228 &      0.053 & Away & 0-3 &  Defeat \\
               Metz &        17 &      0.348 &      0.348 &      0.130 &      0.174 & Home & 1-0 &     Win \\
             Nantes &        18 &      0.382 &      0.345 &      0.109 &      0.164 & Away & 0-0 &    Draw \\
               Lyon &        19 &      0.468 &      0.202 &      0.170 &      0.160 & Home & 2-2 &    Draw \\
     Stade Brestois &        20 &      0.361 &      0.344 &      0.115 &      0.180 & Home & 2-1 &     Win \\
              Lille &        21 &      0.534 &      0.260 &      0.096 &      0.110 & Away & 0-1 &  Defeat \\
          Marseille &        22 &      0.362 &      0.435 &      0.072 &      0.130 & Away & 0-1 &  Defeat \\
            Lorient &        23 &      0.110 &      0.507 &      0.096 &      0.288 & Home & 1-1 &    Draw \\
               Lens &        24 &      0.347 &      0.307 &      0.107 &      0.240 & Away & 0-0 &    Draw \\
      Saint-Étienne &        25 &      0.303 &      0.348 &      0.182 &      0.167 & Home & 0-2 &  Defeat \\
        Montpellier &        26 &      0.333 &      0.300 &      0.150 &      0.217 & Away & 1-2 &  Defeat \\
           OGC Nice &        27 &      0.495 &      0.273 &      0.121 &      0.111 & Home & 1-2 &  Defeat \\
               Lyon &        28 &      0.256 &      0.397 &      0.179 &      0.167 & Away & 0-1 &  Defeat \\
         Strasbourg &        29 &      0.262 &      0.492 &      0.077 &      0.169 & Away & 1-0 &     Win \\
               Metz &        30 &      0.390 &      0.254 &      0.169 &      0.186 & Away & 3-1 &     Win \\
     Stade de Reims &        31 &      0.238 &      0.429 &      0.079 &      0.254 & Away & 2-2 &    Draw \\
             Nantes &        32 &      0.343 &      0.343 &      0.100 &      0.214 & Home & 1-0 &     Win \\
             Angers &        33 &      0.333 &      0.458 &      0.042 &      0.167 & Home & 3-0 &     Win \\
              Dijon &        34 &      0.388 &      0.224 &      0.119 &      0.269 & Away & 5-1 &     Win \\
           Bordeaux &        35 &      0.306 &      0.242 &      0.194 &      0.258 & Home & 0-1 &  Defeat \\
Paris Saint-Germain &        36 &      0.223 &      0.420 &      0.188 &      0.170 & Away & 1-1 &    Draw \\
          AS Monaco &        37 &      0.318 &      0.333 &      0.136 &      0.212 & Home & 1-2 &  Defeat \\
              Nîmes &        38 &      0.477 &      0.138 &      0.092 &      0.292 & Away & 2-0 &     Win \\
\hline
\end{tabular}%
}
\caption{
Proportion of possession sequences across latent clusters for each team and matchday. 
Each row summarizes a team’s cluster profile for a specific fixture, including match location, final score, and outcome from Rennes' perspective. The table illustrates tactical variability across matches, with certain possession styles emerging more frequently in wins or defeats.
}
\label{tab:clusters_wide}
\end{table}
\end{appendices}

% Bibliographie
\bibliography{reference}     % sans le .bib, ici ton fichier s'appelle references.bib

\begin{thebibliography}{31}
\providecommand{\natexlab}[1]{#1}
\providecommand{\url}[1]{\texttt{#1}}
\expandafter\ifx\csname urlstyle\endcsname\relax
  \providecommand{\doi}[1]{doi: #1}\else
  \providecommand{\doi}{doi: \begingroup \urlstyle{rm}\Url}\fi

\bibitem[Baio and Blangiardo(2010)]{baio2010bayesian}
G.~Baio and M.~Blangiardo.
\newblock Bayesian hierarchical model for the prediction of football results.
\newblock \emph{Journal of Applied Statistics}, 37\penalty0 (2):\penalty0
  253--264, 2010.

\bibitem[Baudry et~al.(2010)Baudry, Raftery, Celeux, Lo, and
  Gottardo]{baudry2010combining}
J.-P. Baudry, A.~E. Raftery, G.~Celeux, K.~Lo, and R.~Gottardo.
\newblock Combining mixture components for clustering.
\newblock \emph{Journal of computational and graphical statistics}, 19\penalty0
  (2):\penalty0 332--353, 2010.

\bibitem[Bouvet et~al.(2024)Bouvet, El~Kolei, and
  Marbac]{bouvet2024investigating}
A.~Bouvet, S.~El~Kolei, and M.~Marbac.
\newblock Investigating swimming technical skills by a double partition
  clustering of multivariate functional data allowing for dimension selection.
\newblock \emph{The Annals of Applied Statistics}, 18\penalty0 (2):\penalty0
  1750--1772, 2024.

\bibitem[Cariati et~al.(2025)Cariati, Bonanni, Cifelli, D'Arcangelo, Padua,
  Annino, and Tancredi]{cariati2025virtual}
I.~Cariati, R.~Bonanni, P.~Cifelli, G.~D'Arcangelo, E.~Padua, G.~Annino, and
  V.~Tancredi.
\newblock Virtual reality and sports performance: a systematic review of
  randomized controlled trials exploring balance.
\newblock \emph{Frontiers in Sports and Active Living}, 7:\penalty0 1497161,
  2025.

\bibitem[Demeco et~al.(2024)Demeco, Salerno, Gusai, Vignali, Gramigna, Palumbo,
  Corradi, Mickeviciute, and Costantino]{demeco2024role}
A.~Demeco, A.~Salerno, M.~Gusai, B.~Vignali, V.~Gramigna, A.~Palumbo,
  A.~Corradi, G.~C. Mickeviciute, and C.~Costantino.
\newblock The role of virtual reality in the management of football injuries.
\newblock \emph{Medicina}, 60\penalty0 (6):\penalty0 1000, 2024.

\bibitem[Dempster et~al.(1977)Dempster, Laird, and Rubin]{dempster1977maximum}
A.~P. Dempster, N.~M. Laird, and D.~B. Rubin.
\newblock Maximum likelihood from incomplete data via the {EM} algorithm.
\newblock \emph{Journal of the royal statistical society: series B
  (methodological)}, 39\penalty0 (1):\penalty0 1--22, 1977.

\bibitem[Dixon and Coles(1997)]{dixon1997modelling}
M.~J. Dixon and S.~G. Coles.
\newblock Modelling association football scores and inefficiencies in the
  football betting market.
\newblock \emph{Journal of the Royal Statistical Society: Series C (Applied
  Statistics)}, 46\penalty0 (2):\penalty0 265--280, 1997.

\bibitem[El-Nasr et~al.(2021)El-Nasr, Nguyen, Canossa, and Drachen]{el2021game}
M.~S. El-Nasr, T.-H.~D. Nguyen, A.~Canossa, and A.~Drachen.
\newblock \emph{Game data science}.
\newblock Oxford University Press, 2021.

\bibitem[Fern{\'a}ndez~Mart{\'\i}nez et~al.(2024)Fern{\'a}ndez~Mart{\'\i}nez,
  Casals~Toquero, Oliver, Plensa, and Manisera]{fernandez2024reporting}
D.~Fern{\'a}ndez~Mart{\'\i}nez, M.~Casals~Toquero, M.~Oliver, M.~Plensa, and
  M.~Manisera.
\newblock Reporting of clustering techniques in sports sciences: a scoping
  review.
\newblock \emph{Electronic journal of applied statistical analysis},
  17\penalty0 (3):\penalty0 653--675, 2024.

\bibitem[Fruhwirth-Schnatter et~al.(2019)Fruhwirth-Schnatter, Celeux, and
  Robert]{Fruhwirth2019handbook}
S.~Fruhwirth-Schnatter, G.~Celeux, and C.~P. Robert.
\newblock \emph{Handbook of mixture analysis}.
\newblock CRC press, 2019.

\bibitem[Hennig(2010)]{hennig2010methods}
C.~Hennig.
\newblock Methods for merging gaussian mixture components.
\newblock \emph{Advances in data analysis and classification}, 4\penalty0
  (1):\penalty0 3--34, 2010.

\bibitem[Hennig(2015)]{hennig2015true}
C.~Hennig.
\newblock What are the true clusters?
\newblock \emph{Pattern Recognition Letters}, 64:\penalty0 53--62, 2015.

\bibitem[Hubert and Arabie(1985)]{hubert1985comparing}
L.~Hubert and P.~Arabie.
\newblock Comparing partitions.
\newblock \emph{Journal of classification}, 2:\penalty0 193--218, 1985.

\bibitem[Jensen et~al.(2009)Jensen, Shirley, and Wyner]{jensen2009bayesball}
S.~T. Jensen, K.~E. Shirley, and A.~J. Wyner.
\newblock Bayesball: A bayesian hierarchical model for evaluating fielding in
  major league baseball.
\newblock \emph{The Annals of Applied Statistics}, pages 491--520, 2009.

\bibitem[Jia et~al.(2024)Jia, Sitthiworachart, and Morris]{jia2024application}
T.~Jia, J.~Sitthiworachart, and J.~Morris.
\newblock Application of simulation technology in football training: A
  systematic review of empirical studies.
\newblock \emph{The Open Sports Sciences Journal}, 17\penalty0 (1), 2024.

\bibitem[Kovalchik(2023)]{kovalchik2023player}
S.~A. Kovalchik.
\newblock Player tracking data in sports.
\newblock \emph{Annual Review of Statistics and Its Application}, 10\penalty0
  (1):\penalty0 677--697, 2023.

\bibitem[Leroy et~al.(2018)Leroy, Marc, Dupas, Rey, and
  Gey]{leroy2018functional}
A.~Leroy, A.~Marc, O.~Dupas, J.~L. Rey, and S.~Gey.
\newblock Functional data analysis in sport science: Example of swimmers’
  progression curves clustering.
\newblock \emph{Applied Sciences}, 8\penalty0 (10):\penalty0 1766, 2018.

\bibitem[McLachlan and Peel(2000)]{McL00}
G.~McLachlan and D.~Peel.
\newblock \emph{Finite mixutre models}.
\newblock Wiley Series in Probability and Statistics: Applied Probability and
  Statistics, Wiley-Interscience, New York, 2000.

\bibitem[McLachlan and Krishnan(2008)]{mclachlan2008algorithm}
G.~J. McLachlan and T.~Krishnan.
\newblock \emph{The EM algorithm and extensions}.
\newblock John Wiley \& Sons, 2008.

\bibitem[Michels et~al.(2025)Michels, {\"O}tting, and
  Karlis]{michels2025extending}
R.~Michels, M.~{\"O}tting, and D.~Karlis.
\newblock Extending the dixon and coles model: an application to women’s
  football data.
\newblock \emph{Journal of the Royal Statistical Society Series C: Applied
  Statistics}, 74\penalty0 (1):\penalty0 167--186, 2025.

\bibitem[Narayanan et~al.(2023)Narayanan, Kosmidis, and Dellaportas]{Narayanan}
S.~Narayanan, I.~Kosmidis, and P.~Dellaportas.
\newblock Flexible marked spatio-temporal point processes with applications to
  event sequences from association football.
\newblock \emph{Journal of the Royal Statistical Society Series C: Applied
  Statistics}, 72\penalty0 (5):\penalty0 1095--1126, 08 2023.
\newblock ISSN 0035-9254.
\newblock \doi{10.1093/jrsssc/qlad085}.
\newblock URL \url{https://doi.org/10.1093/jrsssc/qlad085}.

\bibitem[Narizuka and Yamazaki(2019)]{narizuka2019clustering}
T.~Narizuka and Y.~Yamazaki.
\newblock Clustering algorithm for formations in football games.
\newblock \emph{Scientific reports}, 9\penalty0 (1):\penalty0 13172, 2019.

\bibitem[Qi et~al.(2024)Qi, Hu, and Wu]{qi2024made}
K.~Qi, G.~Hu, and W.~Wu.
\newblock Are made and missed different? an analysis of field goal attempts of
  professional basketball players via depth based testing procedure.
\newblock \emph{The Annals of Applied Statistics}, 18\penalty0 (3):\penalty0
  2615--2634, 2024.

\bibitem[Richlan et~al.(2023)Richlan, Wei{\ss}, Kastner, and
  Braid]{richlan2023virtual}
F.~Richlan, M.~Wei{\ss}, P.~Kastner, and J.~Braid.
\newblock Virtual training, real effects: a narrative review on sports
  performance enhancement through interventions in virtual reality.
\newblock \emph{Frontiers in Psychology}, 14:\penalty0 1240790, 2023.

\bibitem[Sandholtz and Bornn(2020)]{Sandholtz}
N.~Sandholtz and L.~Bornn.
\newblock {Markov decision processes with dynamic transition probabilities: An
  analysis of shooting strategies in basketball}.
\newblock \emph{The Annals of Applied Statistics}, 14\penalty0 (3):\penalty0
  1122 -- 1145, 2020.
\newblock \doi{10.1214/20-AOAS1348}.
\newblock URL \url{https://doi.org/10.1214/20-AOAS1348}.

\bibitem[Santos-Fernandez et~al.(2022)Santos-Fernandez, Denti, Mengersen, and
  Mira]{santos2022role}
E.~Santos-Fernandez, F.~Denti, K.~Mengersen, and A.~Mira.
\newblock The role of intrinsic dimension in high-resolution player tracking
  data—insights in basketball.
\newblock \emph{The Annals of Applied Statistics}, 16\penalty0 (1):\penalty0
  326--348, 2022.

\bibitem[Sawczuk et~al.(2024)Sawczuk, Palczewska, Jones, and
  Palczewski]{sawczuk2024bayesian}
T.~Sawczuk, A.~Palczewska, B.~Jones, and J.~Palczewski.
\newblock A bayesian mixture model approach to expected possession values in
  rugby league.
\newblock \emph{PloS one}, 19\penalty0 (11):\penalty0 e0308222, 2024.

\bibitem[Schwarz(1978)]{schwarz1978estimating}
G.~Schwarz.
\newblock Estimating the dimension of a model.
\newblock \emph{The annals of statistics}, pages 461--464, 1978.

\bibitem[Scrucca and Karlis(2025)]{scrucca2025model}
L.~Scrucca and D.~Karlis.
\newblock A model-based approach to shot charts estimation in basketball.
\newblock \emph{Computational Statistics}, 40\penalty0 (4):\penalty0
  2031--2048, 2025.

\bibitem[Witte et~al.(2025)Witte, B{\"u}rger, and Pastel]{witte2025sports}
K.~Witte, D.~B{\"u}rger, and S.~Pastel.
\newblock Sports training in virtual reality with a focus on visual perception:
  a systematic review.
\newblock \emph{Frontiers in Sports and Active Living}, 7:\penalty0 1530948,
  2025.

\bibitem[Yin et~al.(2023)Yin, Hu, and Shen]{yin2023analysis}
F.~Yin, G.~Hu, and W.~Shen.
\newblock Analysis of professional basketball field goal attempts via a
  bayesian matrix clustering approach.
\newblock \emph{Journal of Computational and Graphical Statistics}, 32\penalty0
  (1):\penalty0 49--60, 2023.

\end{thebibliography}
\end{document}